\documentclass{article}

     \PassOptionsToPackage{numbers, compress}{natbib}
 \usepackage[final]{neurips_2021}

\usepackage{amsmath,amsfonts,bm}

\def\eqref#1{equation~\ref{#1}}
\def\1{\bm{1}}

\def\va{{\bm{a}}}

\def\ve{{\bm{e}}}

\def\vh{{\bm{h}}}

\def\vq{{\bm{q}}}

\def\vs{{\bm{s}}}

\def\vu{{\bm{u}}}
\def\vv{{\bm{v}}}

\def\vx{{\bm{x}}}

\def\vz{{\bm{z}}}

\def\evv{{v}}

\def\evx{{x}}

\def\mB{{\bm{B}}}
\def\mC{{\bm{C}}}
\def\mD{{\bm{D}}}

\def\mI{{\bm{I}}}

\def\mM{{\bm{M}}}

\def\mP{{\bm{P}}}

\def\mU{{\bm{U}}}

\DeclareMathAlphabet{\mathsfit}{\encodingdefault}{\sfdefault}{m}{sl}
\SetMathAlphabet{\mathsfit}{bold}{\encodingdefault}{\sfdefault}{bx}{n}

\newcommand{\E}{\mathbb{E}}

\newcommand{\R}{\mathbb{R}}

\usepackage[utf8]{inputenc} % allow utf-8 input
\usepackage[T1]{fontenc}    % use 8-bit T1 fonts
\usepackage{hyperref}       % hyperlinks
\usepackage{url}            % simple URL typesetting
\usepackage{booktabs}       % professional-quality tables
\usepackage{amsfonts}       % blackboard math symbols
\usepackage{nicefrac}       % compact symbols for 1/2, etc.
\usepackage{microtype}      % microtypography
\usepackage{xcolor}         % colors
\usepackage{multirow}
\usepackage{wrapfig}
\usepackage{graphicx}
\usepackage{subcaption}
\usepackage{amssymb}
\usepackage{pifont}
\usepackage{mathptmx}
\usepackage{caption}
\usepackage{mathtools}
\usepackage{amsthm}

\title{On Path Integration of Grid Cells: \\Group Representation and Isotropic Scaling}

\author{
    Ruiqi Gao$^1$\thanks{The author is now a Research Scientist at Google Brain team.}
 \\
    \texttt{ruiqigao@ucla.edu}
    \And Jianwen Xie$^2$ \\
    \texttt{jianwen@ucla.edu}
    \And Xue-Xin Wei$^3$ \\
    \texttt{weixx@utexas.edu}\\
    \And Song-Chun Zhu$^{1, 4, 5}$ \\
    \texttt{sczhu@stat.ucla.edu}
     \And Ying Nian Wu$^1$ \\
    \texttt{ywu@stat.ucla.edu} 
    \And
    \normalfont{$^1$Department of Statistics, UCLA\quad{}}
    \normalfont{$^2$Cognitive Computing Lab, Baidu Research} \\
    \normalfont{$^3$Department of Neuroscience, UT Austin\quad{}}
      \normalfont{$^4$Department of Computer Science, UCLA}\\
    $^5$Beijing Institute for General Artificial Intelligence (BIGAI)}
\newtheorem{prop}{Proposition}
\newtheorem{cond}{Condition}
\newtheorem{theorem}{Theorem}

\begin{document}

\maketitle

\begin{abstract}

Understanding how grid cells perform path integration calculations remains a fundamental problem. In this paper, we conduct theoretical analysis of a general representation model of path integration by grid cells, where the 2D self-position is encoded as a higher dimensional vector, and the 2D self-motion is represented by a general transformation of the vector. We identify two conditions on the transformation. One is a group representation condition that is necessary for path integration. The other is an isotropic scaling condition that ensures locally conformal embedding, so that the error in the vector representation translates conformally to the error in the 2D self-position. Then we investigate the simplest transformation, i.e., the linear transformation, uncover its explicit algebraic and geometric structure as matrix Lie group of rotation, and explore the connection between the isotropic scaling condition and a special class of hexagon grid patterns. Finally, with our optimization-based approach, we manage to learn hexagon grid patterns that share similar properties of the grid cells in the rodent brain. The learned model is capable of accurate long distance path integration. Code is available at \url{https://github.com/ruiqigao/grid-cell-path}.

\end{abstract}

\section{Introduction}

Imagine walking in the darkness. Purely based on the sense of self-motion, one can gain a sense of self-position by integrating the self motion - a process often referred to as path integration~\citep{Darwin1873,Etienne2004,hafting2005microstructure, fiete2008grid, mcnaughton2006path}. While the exact neural underpinning of path integration remains unclear, it has been hypothesized that the grid cells~\citep{hafting2005microstructure, fyhn2008grid, yartsev2011grid, killian2012map, jacobs2013direct, doeller2010evidence} in the mammalian medial entorhinal cortex (mEC) may be involved in this process~\citep{gil2018impaired, ridler2019impaired, horner2016grid}. 
The grid cells are so named because individual neurons exhibit striking firing patterns that form  hexagonal grids when the agent (such as a rat) navigates in a 2D open field ~\citep{Fyhn2004,hafting2005microstructure,Fuhs2006,burak2009accurate, sreenivasan2011grid, blair2007scale,Couey2013, de2009input,Pastoll2013,Agmon2020}.
The grid cells also interact with the place cells in the hippocampus~\citep{o1979review}. Unlike a grid cell that fires at the vertices of a lattice, a place cell often fires at a single (or a few) locations.

The purpose of this paper is to understand how the grid cells may perform path integration calculations. We study a general optimization-based representational model in which the 2D self-position is represented by a higher dimensional vector and the 2D self-motion is represented by a transformation of the vector. The vector representation can be considered position encoding or position embedding, where the elements of the vector may be interpreted as activities of a population of grid cells. The transformation can be realized by a recurrent network that acts on the vector. Our focus is to study the properties of the transformation. 

Specifically, we identify two conditions for the transformation: a group representation condition and an isotropic scaling condition, under which we demonstrate that the local neighborhood around each self-position in the 2D physical space is embedded conformally as a 2D neighborhood around the vector representation of the self-position in the neural space.

We then investigate the simplest special case of the transformation, i.e., linear transformation,  that forms a matrix Lie group of rotation, under which case we show that the isotropic scaling condition is connected to a special class of hexagonal grid patterns. Our numerical experiments demonstrate that our model learns clear hexagon grid patterns of multiple scales which share observed properties of the grid cells in the rodent brain, by optimizing a simple loss function. The learned model is also capable of  accurate long distance path integration.  

{\bf Contributions}. Our work contributes to understanding the grid cells from the perspective of representation learning. We conduct theoretical analysis of (1) general transformation for path integration by identifying two key conditions and a local conformal embedding property, (2) linear transformation by revealing the algebraic and geometric structure and connecting the isotropic scaling condition and a special class of hexagon grid patterns, and (3) integration of linear transformation model and linear basis expansion model via unitary group representation theory.  Experimentally we learn clear hexagon grid patterns that are consistent with biological observations, and the learned model is capable of accurate path integration.

\section{General transformation} 
 \begin{wrapfigure}{r}{0.5\linewidth}  
%\begin{figure}
	\centering	
	\begin{tabular}{cc}
\includegraphics[width=.44\linewidth]{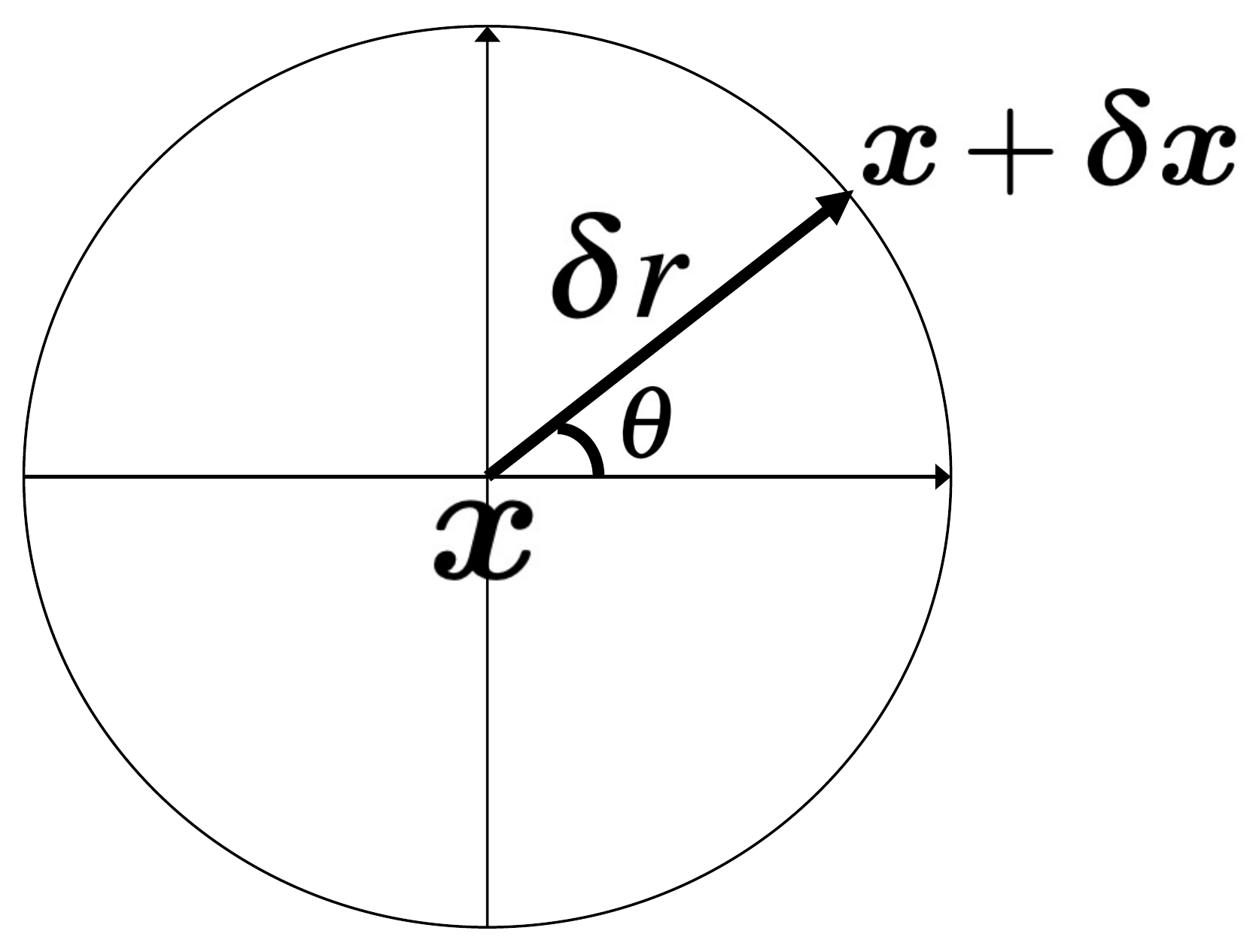}  &
		\includegraphics[width=.5\linewidth]{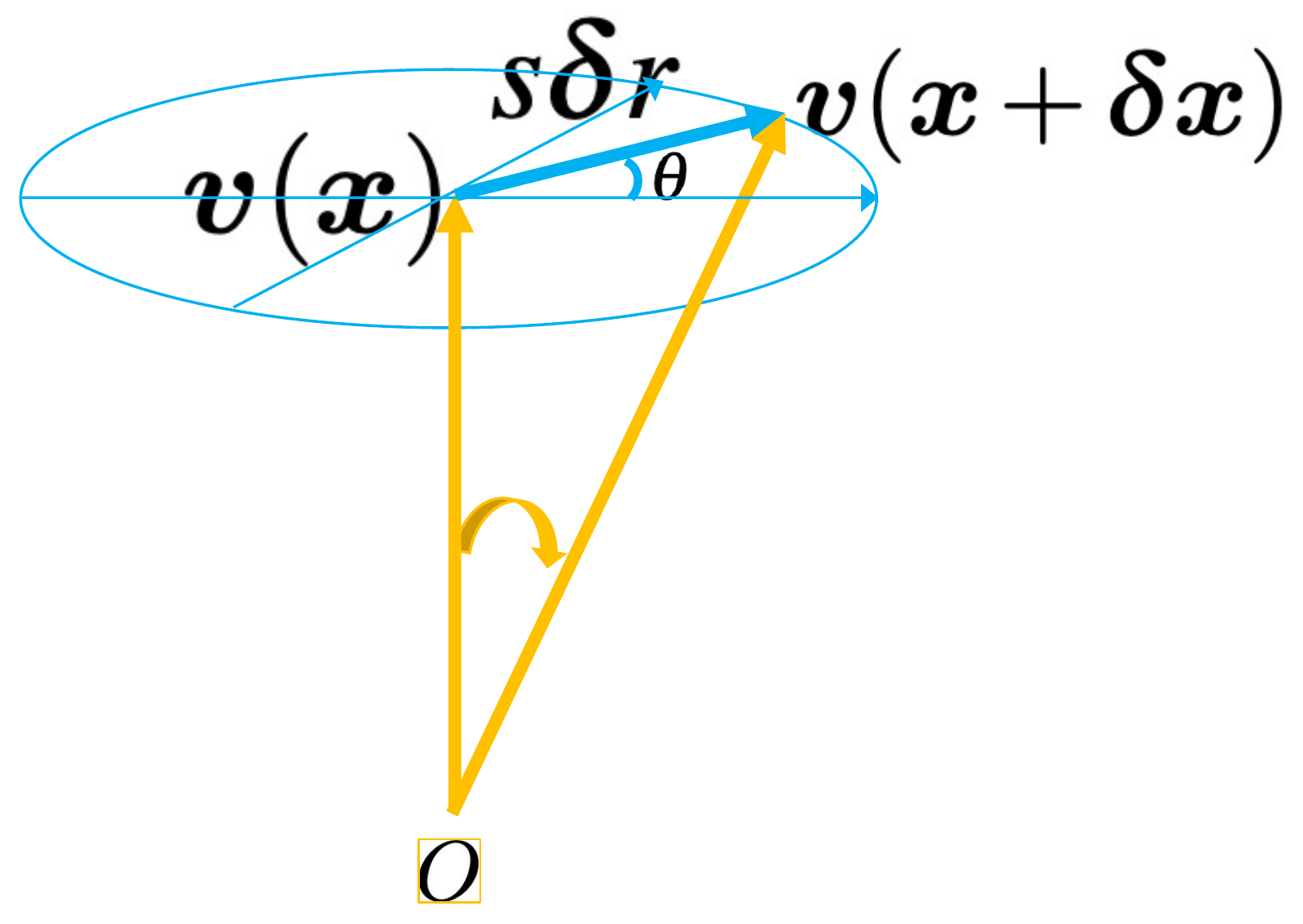}  \\
		{\small (a) physical space} & {\small (b) neural space}  
							\end{tabular} 						
	\caption{\small The local 2D polar system around self-position $\vx$ in the 2D physical space (a) is embedded conformally as a 2D polar system around vector $\vv(\vx)$ in the $d$-dimensional neural space (b), with a scaling factor $s$ (so that $\delta r$ in the physical space becomes $s \delta r$ in the neural space while the angle $\theta$ is preserved).  
	}  
	\vspace{-.4cm}
	\label{fig:g}
%\end{figure}
\end{wrapfigure} 

\subsection{Position embedding} 

Consider an agent (e.g., a rat) navigating within a 2D open field.  Let $\vx = (\evx_1, \evx_2)$ be the self-position of the agent. We assume that the self-position $\vx$ in the 2D physical space is represented by the response activities of a population of $d$ neurons (e.g., $d = 200$), which form a vector $\vv(\vx) = (\evv_i(\vx), i = 1, ..., d)^\top$ in the $d$-dimensional ``neural space'', with each element $\evv_i(\vx)$ representing the firing rate of one neuron when the animal is at location $\vx$. 

$\vv(\vx)$ can be called position encoding or position embedding. Collectively, $(\vv(\vx), \forall \vx)$ forms a {\em codebook} of $\vx \in \R^2$, and $(\vv(\vx), \forall \vx)$ is a {\em 2D manifold} in the $d$-dimensional neural space, i.e., globally we embed $\R^2$ as a 2D manifold in the neural space. Locally, we identify two conditions under which the 2D local neighborhood around each $\vx$ is embedded {\em conformally} as a 2D neighborhood around $\vv(\vx)$ with a scaling factor. See Fig. \ref{fig:g}. As shown in Section \ref{sect:hexagon}, the conformal embedding is connected to the hexagon grid patterns. 

\subsection{Transformation and path integration}

At self-position $\vx$,  if the agent makes a self-motion $\Delta \vx = (\Delta \evx_1, \Delta \evx_2)$, then it moves to $\vx + \Delta \vx$. Correspondingly, the vector representation $\vv(\vx)$ is transformed to $\vv(\vx+\Delta \vx)$. The general form of the transformation can be formulated as: 
\begin{eqnarray} 
    \vv(\vx + \Delta \vx)  = F(\vv(\vx), \Delta \vx). \label{eq:PI0}
\end{eqnarray}
The transformation $F(\cdot, \Delta \vx)$ can be considered a representation of $\Delta \vx$, which forms a 2D additive group.  We call Eq. (\ref{eq:PI0}) the {\em transformation model}. It can be implemented by a recurrent network to derive a path integration model: if the agent starts from $\vx_0$, and makes a sequence of moves $(\Delta \vx_t, t = 1, ..., T)$, then the vector is updated by  $\vv_{t} = F(\vv_{t-1}, \Delta \vx_t)$, where $\vv_0 = \vv(\vx_0)$, and $t = 1, ..., T$. 

\subsection{Group representation condition}  

The solution to the transformation model (Eq. (\ref{eq:PI0})) should satisfy the following condition. 
\begin{cond} \label{cond:1}
(Group representation condition) $(\vv(\vx), \forall \vx)$ and $(F(\cdot, \Delta \vx), \forall \Delta \vx)$ form a representation of the 2D additive Euclidean group $\R^2$ in the sense that  
\begin{eqnarray} 
   &&F(\vv(\vx), 0) = \vv(\vx), \; \forall \vx;  \label{eq:attract} \\
   &&F(\vv(\vx), \Delta \vx_1 + \Delta \vx_2) = F(F(\vv(\vx), \Delta \vx_1), \Delta \vx_2), \; \forall \vx, \Delta \vx_1, \Delta \vx_2. \label{eq:gr}
\end{eqnarray} 
\end{cond}
$(F(\cdot, \Delta \vx), \forall \Delta \vx)$ is a Lie group of transformations acting on the codebook manifold $(\vv(\vx), \forall \vx)$. 

The reason for (\ref{eq:attract}) is that if $\Delta \vx = 0$, then $F(\cdot, 0)$ should be the identity transformation. Thus the codebook manifold $(\vv(\vx), \forall \vx)$ consists of fixed points of the transformation $F(\cdot, 0)$. If $F(\cdot, 0)$ is furthermore a contraction around $(\vv(\vx), \forall \vx)$, then $(\vv(\vx), \forall \vx)$ are the attractor points. 

The reason for (\ref{eq:gr}) is that the agent can move in one step by $\Delta \vx_1 + \Delta \vx_2$, or first move by $\Delta \vx_1$, and then move by $\Delta \vx_2$. Both paths would end up at the same $\vx + \Delta \vx_1 + \Delta \vx_2$, which is represented by the same $\vv(\vx+\Delta \vx_1 + \Delta \vx_2)$. 

The group representation condition is a necessary self-consistent condition for the transformation model (Eq. (\ref{eq:PI0})). 
   
\subsection{Egocentric self-motion} 

Self-motion $\Delta \vx$ can also be parametrized egocentrically as $(\Delta r, \theta)$, where $\Delta r$ is the displacement along the direction $\theta \in [0, 2\pi]$, so that $\Delta \vx = (\Delta \evx_1 = \Delta r \cos \theta, \Delta \evx_2 = \Delta r \sin \theta)$. The egocentric self-motion may be more biologically plausible where $\theta$ is encoded by head direction, and $\Delta r$ can be interpreted as the speed along direction $\theta$. The transformation model  then becomes 
\begin{eqnarray} 
   \vv(\vx + \Delta \vx) = F(\vv(\vx), \Delta r, \theta),  \label{eq:PI2}
\end{eqnarray}
where we continue to use $F(\cdot)$ for the transformation (with slight abuse of notation). $(\Delta r, \theta)$ form a polar coordinate system around $\vx$. 

\subsection{Infinitesimal self-motion and directional derivative}  
 
In this subsection, we derive the transformation model for infinitesimal self-motion. While we use $\Delta \vx$ or $\Delta r$ to denote finite (non-infinitesimal) self-motion, we use $\delta \vx$ or $\delta r$ to denote infinitesimal self-motion. At self-position $\vx$, for an infinitesimal displacement $\delta r$ along direction $\theta$, $\delta \vx = (\delta \evx_1 = \delta r \cos \theta, \delta \evx_2 = \delta r \sin \theta)$. See Fig. \ref{fig:g} (a) for an illustration. Given that $\delta r$ is infinitesimal, for any fixed $\theta$, a first order Taylor expansion of $F(\vv(\vx), \delta r, \theta)$ with respect to $\delta r$  gives us 
\begin{align}
\vv(\vx + \delta \vx) &=F(\vv(\vx), \delta r, \theta) = F(\vv(\vx), 0, \theta) + F'(\vv(\vx), 0, \theta) \delta r + o(\delta r) \nonumber 
\\&= \vv(\vx) + f_\theta(\vv(\vx)) \delta r + o(\delta r), 
\label{eq:inf}
\end{align} 
 where $F(\vv(\vx), 0, \theta) = \vv(\vx)$ according to Condition \ref{cond:1}, and $f_\theta(\vv(\vx)) \coloneqq F'(\vv(\vx), 0, \theta)$ is the first derivative of $F(\vv(\vx), \Delta r, \theta)$ with respect to $\Delta r$ at $\Delta r = 0$. $f_\theta(\vv(\vx))$ is the {\em directional derivative} of $F(\cdot)$ at self-position $\vx$ and direction $\theta$. 
 
 For a fixed $\theta$, $(F(\cdot, \Delta r, \theta), \forall \Delta r)$ forms a one-parameter Lie group of transformations, and $f_\theta(\cdot)$ is the generator of its Lie algebra. 
   
 \subsection{Isotropic scaling condition} 
 
With the directional derivative, we define the second condition as follows, which leads to locally conformal embedding and is connected to hexagon grid pattern. 
 \begin{cond} \label{cond:2}
(Isotropic scaling condition) For any fixed $\vx$, $\|f_\theta(\vv(\vx))\|$ is constant over $\theta$. 
\end{cond}
Let $f_0(\vv(\vx))$ denote $f_\theta(\vv(\vx))$ for $\theta = 0$, and $f_{\pi/2}(\vv(\vx))$ denote $f_\theta(\vv(\vx))$ for $\theta = \pi/2$. Then we have the following theorem: 
 \begin{theorem} \label{theorem:1}
Assume group representation condition \ref{cond:1} and isotropic scaling condition \ref{cond:2}.  At any fixed $\vx$, for the local motion $\delta \vx = (\delta r \cos \theta,  \delta r \sin \theta)$ around $\vx$, let $\delta \vv = \vv(\vx+\delta \vx) - \vv(\vx)$ be the change of vector and $s = \|f_\theta(\vv(\vx))\|$, then we have $\|\delta \vv\| = s \|\delta \vx\|$. Moreover,
\begin{align}
\delta \vv  
= f_\theta(\vv(\vx)) \delta r + o(\delta r)= f_0(\vv(\vx)) \delta r \cos \theta + f_{\pi/2}(\vv(\vx)) \delta r \sin \theta + o(\delta r),
\end{align}
where $f_0(\vv(\vx))$ and  $f_{\pi/2}(\vv(\vx))$ are two orthogonal basis vectors of equal norm $s$. 
\end{theorem}
See Supplementary for a proof and Fig. \ref{fig:g}(b) for an illustration. Theorem \ref{theorem:1} indicates that the local 2D polar system around self-position $\vx$ in the 2D physical space is embedded conformally as a 2D polar system around vector $\vv(\vx)$ in the $d$-dimensional neural space, with a scaling factor $s$ (our analysis is local for any fixed $\vx$, and $s$ may depend on $\vx$). Conformal embedding is a generalization of isometric embedding, where the metric can be changed by a scaling factor $s$. If $s$ is globally constant for all $\vx$, then the intrinsic geometry of the codebook manifold $(\vv(\vx), \forall \vx)$ remains Euclidean, i.e., flat. 
 
{\bf Why isotropic scaling and conformal embedding?} The neurons are intrinsically noisy. During path integration, the errors may accumulate in $\vv$. Moreover, when inferring self-position from visual image, it is possible that $\vv$ is inferred first with error, and then $\vx$ is decoded from the inferred $\vv$. Due to  isotropic scaling and conformal embedding, locally we have $\|\delta \vv\| = s \|\delta \vx\|$, which guarantees that the $\ell_2$ error in $\vv$ translates proportionally to the $\ell_2$ error in $\vx$, so that there will not be adversarial perturbations in $\vv(\vx)$ that cause excessively big errors in $\vx$.  Specifically, we have the following theorem. 

\begin{theorem} \label{theorem:a}
Assume the general transformation model (Eq. (\ref{eq:PI2})) and the isotropic scaling condition. For any fixed $\vx$, let $s =   \|f_\theta(\vv(\vx))\|$, which is independent of $\theta$. Suppose the neurons are noisy: $\vv = \vv(\vx) + {\bf \epsilon}$, where ${\bf \epsilon} \sim \mathcal{N}(0, \tau^2 \mI_d)$ and $d$ is the dimensionality of $\vv$. Suppose the agent infers its 2D position $\hat{\vx}$ from $\vv$ by $\hat{\vx} = \arg\min_{\vx'}\|\vv - \vv(\vx')\|^2$, i.e., $\vv(\hat{\vx})$ is the projection of $\vv$ onto the 2D manifold formed by $(\vv(\vx'), \forall \vx')$. Then we have 
\begin{align}
\E\|\hat{\vx} - \vx\|^2 = 2 \tau^2/s^2.
\end{align}
\end{theorem} 
See Supplementary for a proof. 

{\bf Connection to continuous attractor neural network (CANN) defined on 2D torus}. The group representation condition and the isotropic scaling condition appear to be satisfied by the CANN models 
\citep{amit1992modeling,burak2009accurate,Couey2013,Pastoll2013,Agmon2020} that are typically hand-designed on a 2D torus. See Supplementary for details.

 \section{Linear transformation} 
 
 After studying the general transformation, we now investigate the linear transformation of $\vv(\vx)$, for the following reasons. (1) It is the simplest transformation for which we can derive  explicit algebraic and geometric results. (2) It enables us to connect the isotropic scaling condition to a special class of hexagon grid patterns. (3) In Section \ref{sect:i}, we integrate it with the basis expansion model, which is also linear in $\vv(\vx)$, via unitary group representation theory. 
 
 For finite (non-infinitesimal) self-motion, the linear transformation model is: 
 \begin{eqnarray} \label{eqn:linear}
 	\vv(\vx + \Delta \vx)  = F(\vv(\vx), \Delta \vx) = \mM(\Delta \vx) \vv(\vx),
 \end{eqnarray}
where $\mM(\Delta \vx)$ is a matrix. The group representation condition becomes $\mM(\Delta \vx_1 + \Delta \vx_2)  \vv(\vx) = \mM(\Delta \vx_2) \mM(\Delta \vx_1) \vv(\vx)$, i.e., $\mM(\Delta \vx)$ is a matrix representation of self-motion $\Delta \vx$, and $\mM(\Delta \vx)$ acts on the coding manifold $(\vv(\vx), \forall \vx)$). For egocentric parametrization of self-motion $(\Delta r, \theta)$, we can further write $\mM(\Delta \vx) = \mM_\theta(\Delta r)$ for $\Delta \vx = (\Delta r \cos \theta, \Delta r \sin \theta)$, and the linear model becomes $\vv(\vx + \Delta \vx)  = F(\vv(\vx), \Delta r, \theta) = \mM_\theta(\Delta r) \vv(\vx)$.
  
 \subsection{Algebraic structure: matrix Lie algebra and Lie group} \label{sect:algebraic}
For the linear model (Eq. (\ref{eqn:linear})), the directional derivative is: $f_\theta(\vv(\vx)) = F'(\vv(\vx), 0, \theta) = \mM'_\theta(0) \vv(\vx) = \mB(\theta) \vv(\vx)$, where $\mB(\theta) = \mM'_\theta(0)$, which is the derivative of $\mM_\theta(\Delta r)$ with respect to $\Delta r$ at $0$.
For infinitesimal self-motion, the transformation model in Eq. (\ref{eq:inf}) becomes
\begin{eqnarray}
   \vv(\vx+\delta \vx) =  (\mI + \mB(\theta) \delta r) \vv(\vx) + o(\delta r), \label{eq:a1}
\end{eqnarray}
 where $\mI$ is the identity matrix. It can be considered a linear recurrent network where $\mB(\theta)$ is the learnable weight matrix. We have the following theorem for the algebraic structure of the linear transformation.
 
  \begin{theorem} \label{theorem:2}
 Assume the linear transformation model so that for infinitesimal self-motion $(\delta r, \theta)$, the model is in the form of Eq. (\ref{eq:a1}), then for finite displacement $\Delta r$, 
 \begin{equation} \label{eq:linear2}
\vv(\vx + \Delta \vx) = \mM_\theta(\Delta r) \vv(\vx) = \exp(\mB(\theta) \Delta r) \vv(\vx).  
\end{equation}
 \end{theorem} 
\begin{proof} 
	We can divide $\Delta r$ into $N$ steps, so that $\delta r = \Delta r/N \rightarrow 0$ as $N \rightarrow \infty$, and 
\begin{align} \label{eq:1} 
 \vv(\vx+\Delta \vx) &= (\mI + \mB(\theta) (\Delta r/N) + o(1/N))^N  \vv(\vx)
 \rightarrow  \exp(\mB(\theta) \Delta r) \vv(\vx)
\end{align}
as $N \rightarrow \infty$. The matrix exponential map is defined by $\exp(A) = \sum_{n=0}^{\infty} A^{n}/n!$.  
\end{proof}
The above math underlies the relationship between matrix Lie algebra and matrix Lie group in general~\citep{taylor2002lectures}. For a fixed $\theta$, the set of $\mM_\theta(\Delta r) = \exp(\mB(\theta) \Delta r)$ for $\Delta r \in \R$ forms a {\em matrix Lie group}, which is both a group and a manifold. The tangent space of $\mM_\theta(\Delta r)$ at identity $\mI$ is called {\em matrix Lie algebra}. $\mB(\theta)$ is the basis of this tangent space, and is often referred to as the {\em generator matrix}.

{\bf Path integration}.  If the agent starts from $\vx_0$, and make a sequence of moves $((\Delta r_t, \theta_t),  t = 1, ..., T)$, then the vector representation of self-position is updated by 
\begin{eqnarray} 
    \vv_{t} = \exp(\mB(\theta_t) \Delta r_t) \vv_{t-1}, \label{eq:PI}
\end{eqnarray} 
where $\vv_0 = \vv(\vx_0)$, and $t = 1, ..., T$. 

{\bf Approximation to exponential map}. For a finite but small $\Delta r$,  $\exp(\mB(\theta) \Delta r)$ can be approximated by a second-order (or higher-order) Taylor expansion
\begin{eqnarray}
\exp(\mB(\theta)\Delta r) = \mI + \mB(\theta)\Delta r + \mB(\theta)^2 \Delta r^2/2 + o(\Delta r^2). \label{eq:second}
\end{eqnarray}

\subsection{Geometric structure: rotation, periodicity, metic and error correction} \label{sect:r}

If we assume $\mB(\theta) = - \mB(\theta)^{\top}$, i.e., skew-symmetric, then $\mI + \mB(\theta) \delta r$ in Eq. (\ref{eq:a1}) is  a rotation matrix operating on $\vv(\vx)$, due to the fact that $(\mI + \mB(\theta) \delta r)(\mI + \mB(\theta) \delta r)^{\top} = \mI + O(\delta r^2)$. For finite $\Delta r$, $\exp(\mB(\theta) \Delta r)$ is also  a rotation matrix, as it equals to the product of $N$ matrices $\mI + \mB(\theta) (\Delta r/N)$ (Eq. (\ref{eq:1})). 
The geometric interpretation is that, if the agent moves along the direction $\theta$ in the physical space, the vector $\vv(\vx)$ is rotated by the matrix $\mB(\theta)$ in the neural space, while the $\ell_2$ norm $\|\vv(\vx)\|^2$ remains fixed.  We may interpret $\|\vv(\vx)\|^2 = \sum_{i=1}^{d} v_i(\vx)^2$ as the total energy of grid cells. 
See Fig. \ref{fig:g}(b).

The angle of rotation is given by $\|\mB(\theta)  \vv(\vx)\|\delta r / \|\vv(\vx)\|$, because $\|\mB(\theta)  \vv(\vx)\|\delta r$ is the arc length and $\|\vv(\vx)\|$ is the radius. If we further assume the isotropic scaling condition, which becomes that $\|f_\theta(\vv(\vx))\| = \|\mB(\theta) \vv(\vx)\|$ is constant over $\theta$ for the linear model, then the angle of rotation can be written as $\mu \delta r$, where  $\mu =   \|\mB(\theta) \vv(\vx)\|/\|\vv(\vx)\|$ is independent of $\theta$.
 Geometrically, $\mu$ tells us how fast the vector rotates in the neural space as the agent moves in the physical space. In practice, $\mu$ can be much bigger than 1 for the learned model, thus the vector can rotate back to itself in a short distance, causing the periodic patterns in the elements of $\vv(\vx)$. $\mu$ captures the notion of metric. 
 
For $\mu \gg 1$, the conformal embedding  in Fig. \ref{fig:g} (b) {\bf magnifies} the local motion in Fig. \ref{fig:g} (a), and this enables error correction~\citep{sreenivasan2011grid}. More specifically, we have the following result, which is based on Theorem \ref{theorem:a}.  

\begin{prop} \label{prop:a}
Assume the linear transformation model (Eq. (\ref{eq:a1})) and the isotropic scaling condition \ref{cond:2}. For any fixed $\vx$, let $\mu =   \|\mB(\theta) \vv(\vx)\|/\|\vv(\vx)\|$. Suppose  $\vv = \vv(\vx) + {\bf \epsilon}$, where ${\bf \epsilon} \sim \mathcal{N}(0, \tau^2 \mI_d)$ and $\tau^2 = \alpha^2 (\|\vv(\vx)\|^2/d)$, so that $\alpha^2$ measures the variance of noise relative to the average magnitude of $(v_i(\vx)^2, i = 1, ..., d)$. Suppose the agent infers its 2D position $\hat{\vx}$ from $\vv$ by $\hat{\vx} = \arg\min_{\vx'}\|\vv - \vv(\vx')\|^2$. Then we have 
\begin{align}
\E\|\hat{\vx} - \vx\|^2 = 2 \alpha^2/(\mu^2 d).
\end{align}
\end{prop}
See Supplementary for a proof. By the above proposition, error correction of grid cells is due to two factors: (1) higher dimensionality $d$ of $\vv(\vx)$ for encoding 2D positions $\vx$, and (2) a magnifying $\mu \gg 1$ (our analysis is local for any fixed $\vx$, and $\mu$ may depend on $\vx$). 

\subsection{Hexagon grid patterns formed by mixing Fourier waves} \label{sect:hexagon}

In this subsection, we make connection between the isotropic scaling condition \ref{cond:2} and a special class of hexagon grid patterns created by linearly mixing three Fourier plane waves whose directions are $2\pi/3$ apart. We show such linear mixing satisfies the linear transformation model and the isotropic scaling condition. 

\begin{theorem} \label{theorem:3}
 Let $\ve(\vx) = (\exp(i \langle \va_j,\vx\rangle), j = 1, 2, 3)^\top$, where $(\va_j, j = 1, 2, 3)$ are three 2D vectors of equal norm, and the angle between every pair of them is  $2\pi/3$. Let $\vv(\vx) = \mU \ve(\vx)$, where $\mU$ is an arbitrary unitary matrix. Let $\mB(\theta) = \mU^* \mD(\theta) \mU$,  where $\mD(\theta) = {\rm diag}(i \langle \va_j, \vq(\theta)\rangle, j = 1, 2, 3)$, with $\vq(\theta) = (\cos \theta, \sin \theta)^\top$. Then $(\vv(\vx), \mB(\theta))$ satisfies the linear transformation model (Eq. (\ref{eq:a1})) and the isotropic scaling condition \ref{cond:2}. Moreover, $\mB(\theta)$ is skew-symmetric. 
 \end{theorem}
See Supplementary for a proof. We would like to emphasize that the above theorem analyzes a special case solution to our linear transformation model, but our optimization-based learning method {\bf does not assume any superposition of Fourier basis functions} as in the theorem. Our experimental results are learned purely by optimizing a loss function based on the simple assumptions of our model with generic vectors and matrices. 

We leave it to future work to theoretically prove that the isotropic scaling condition leads to hexagon grid patterns in either the general transformation model or the linear transformation model. The hexagon grid patterns are not limited to superpositions of three plane waves as in the above theorem. 

\subsection{Modules} \label{sect:modules}

Biologically, it is well established that grid cells are organized in discrete modules~\citep{Barry2007experience,stensola2012entorhinal} or blocks. We thus partition the vector $\vv(\vx)$ into $K$ blocks, $\vv(\vx) = (\vv_{k}(\vx), k = 1, ..., K)$. Correspondingly the generator matrices $\mB(\theta) = {\rm diag}(\mB_{k}(\theta), k = 1, ..., K)$  are block diagonal, so that each sub-vector $\vv_{k}(\vx)$ is rotated by a sub-matrix $\mB_{k}(\theta)$. For the general transformation model, each sub-vector is transformed by a separate sub-network. 
By the same argument as in Section \ref{sect:r},  let $\mu_k = \|\mB_{k} \vv_k(\vx)\|/\|\vv_k(\vx)\|$, then $\mu_k$ is the metric of module $k$.

\section{Interaction with place cells} 
\label{sect:i}

\subsection{Place cells} 

For each $\vv(\vx)$, we need to uniquely decode $\vx$ globally. This can be accomplished via interaction with place cells. Specifically, each place cell fires when the agent is at a specific position. Let $A(\vx, \vx')$ be the response map of the place cell associated with position $\vx'$. It measures the adjacency between $\vx$ and $\vx'$. A commonly used form of $A(\vx, \vx')$  is the Gaussian adjacency kernel $A(\vx, \vx') = \exp(-\|\vx-\vx'\|^2/(2\sigma^2))$. The set of Gaussian adjacency kernels serve as inputs to our optimization-based method to learn grid cells.

\subsection{Basis expansion} 
 
 \begin{wrapfigure}{r}{0.38\linewidth} 
%\begin{figure}[h]
\centering
\vspace{-.3cm}
\includegraphics[width=.85\linewidth]{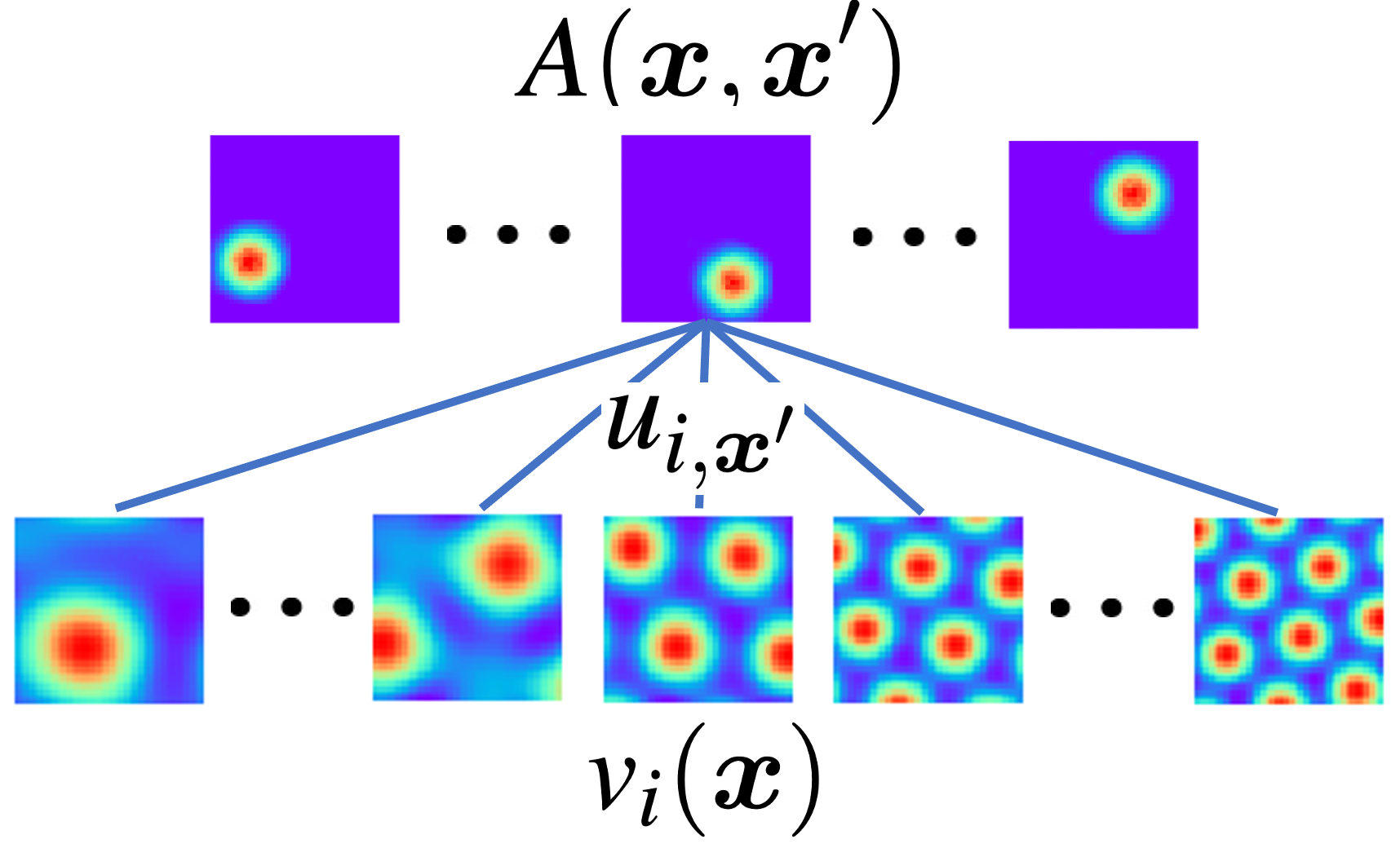}
\caption{\small Illustration of basis expansion model $A(\vx, \vx') = \sum_{i=1}^{d} u_{i,\vx'} v_i(\vx)$, where $v_i(\vx)$ is the response map of $i$-th grid cell, shown at the bottom, which shows 5 different $i$. $A(\vx, \vx')$ is the response map of place cell associated with $\vx'$, shown at the top, which shows 3 different $\vx'$. $u_{i, \vx'}$ is the connection weight.}
\vspace{-1.3cm}
\label{fig:units1}
%\end{figure}
\end{wrapfigure}  
 A popular model that connects place cells and grid cells is the following basis expansion model (or PCA-based model)~\citep{dordek2016extracting}: 
\begin{align} 
A(\vx, \vx') = \langle \vv(\vx), \vu(\vx')\rangle = \sum_{i=1}^{d} u_{i, \vx'} v_i(\vx), \label{eq:b1}
\end{align} 
where $\vv(\vx) = (v_i(\vx), i = 1, ..., d)^\top$, and $\vu(\vx') = (u_{i, \vx'}, i = 1, ..., d)^\top$. Here $(v_i(\vx), i=1, ..., d)$ forms a set of $d$ basis functions (which are functions of $\vx$) for expanding $A(\vx, \vx')$ (which is a function of $\vx$ for each place $\vx'$), while $\vu(\vx')$ is the read-out weight vector for place cell at $\vx'$ and needs to be learned. See Fig. \ref{fig:units1} for an illustration. 
Experimental results on biological brains have shown that the connections from grid cells to place cells are excitatory \citep{zhang2013optogenetic,Rowland2018}. We thus assume that $u_{i, \vx'} \geq 0$ for all $i$ and $\vx'$. 
\subsection{From group representation to basis functions}

The vector representation $\vv(\vx)$ generated (or constrained) by the linear transformation model (Eq. (\ref{eqn:linear})) can serve as basis functions of the PCA-based basis expansion model (Eq. (\ref{eq:b1})), due to the fundamental theorems of Schur~\citep{zee2016group} and Peter-Weyl~\citep{taylor2002lectures}, which reveal the deep root of Fourier analysis and generalize it to general Lie groups. Specifically, if $\mM(\Delta \vx)$ is an irreducible unitary representation of $\Delta \vx$ that forms a compact Lie group, then the elements $\{M_{ij}(\Delta \vx)\}$ form a set of orthogonal basis functions of $\Delta \vx$.  Let $\vv(\vx) = \mM(\vx) \vv(0)$ (where we choose the origin $0$ as the reference point). The elements of $\vv(\vx)$, i.e., $(v_i(\vx), i = 1, ..., d)$, are linear mixings of the basis functions $\{M_{ij}(\vx)\}$, so that they themselves form a new set of basis functions that serve to expand $(A(\vx, \vx'), \forall \vx')$ that parametrizes the place cells. Thus group representation in our path integration model is a perfect match to the basis expansion model, in the sense that the basis functions are results of group representation. 

The basis expansion model (or PCA-based model) (Eq. (\ref{eq:b1})) assumes that the basis functions are orthogonal, whereas in our work, {\bf we do not make the orthogonality assumption}. Interestingly, the learned transformation model generates basis functions that are close to being orthogonal automatically. See Supplementary for more detailed explanation and experimental results. 

\subsection{Decoding and re-encoding}

For a neural response vector $\vv$, such as $\vv_t$ in Eq. (\ref{eq:PI}), the response of the place cell associated with  location $\vx'$ is $\langle \vv, \vu(\vx')\rangle$. We can decode the position $\hat{\vx}$ by examining which place cell has the maximal response, i.e.,
 \begin{eqnarray}
\hat{\vx} = \arg\max_{\vx'} \langle \vv, \vu(\vx')\rangle. \label{eq:decode}
\end{eqnarray}
After decoding $\hat{\vx}$, we can re-encode $\vv \leftarrow \vv(\hat{\vx})$ for error correction. Decoding and re-encoding can also be done by directly projecting $\vv$ onto the manifold $(\vv(\vx), \forall \vx)$, which gives similar results. See Supplementary for more analysis and experimental results. 

\section{Learning} \label{sect:learning}

We learn the model by optimizing a loss function defined based on three model assumptions discussed above: (1) the basis expansion model (Eq. (\ref{eq:b1})),  (2) the linear transformation model (Eq. (\ref{eq:linear2})) and (3) the isotropic scaling condition \ref{cond:2}. The input is the set of adjacency kernels $A(\vx, \vx'), \forall \vx, \vx'$. The unknown parameters to be learned are (1) $(\vv(\vx) = (\vv_k(\vx), k = 1, ..., K), \forall \vx)$, (2) $(\vu(\vx'), \forall \vx' )$ and (3) ($\mB(\theta), \forall \theta)$. We assume that there are $K$ modules or blocks and $\mB(\theta)$ is skew-symmetric, so that $\mB(\theta)$ are parametrized as block-diagonal matrices $(\mB_k(\theta), k = 1, ..., K), \forall \theta)$ and only the lower triangle parts of the matrices need to be learned. The loss function is defined as a weighted sum of simple $\ell_2$ loss terms constraining the three model assumptions: $L = L_0 + \lambda_1 L_1 + \lambda_2 L_2$, where 
\begin{align} 
L_0 &= \E_{\vx, \vx'}[A(\vx, \vx') - \langle \vv(\vx), \vu(\vx')\rangle]^2, \; (\mbox{basis expansion}) \label{eq:L1}\\
L_1 &= \sum_{k=1}^{K} \E_{\vx, \Delta \vx}  \|\vv_k(\vx+\Delta \vx) - \exp(\mB_k(\theta) \Delta r) \vv_k(\vx)\|^2, \label{eq:L2}\; (\mbox{transformation})\\
L_2 &= \sum_{k=1}^{K} \E_{\vx, \theta, \Delta \theta} [\|\mB_k(\theta + \Delta \theta) \vv_k(\vx)\| - \|\mB_k(\theta) \vv_k(\vx)\| ]^2.  \; (\mbox{isotropic scaling}) \label{eq:L3}
\end{align}
In $L_1$, $\Delta \vx = (\Delta r \cos \theta, \Delta r \sin \theta)$. $\lambda_1$ and $\lambda_2$ are chosen so that the three loss terms are of similar magnitudes. $A(\vx, \vx')$ are given as Gaussian adjacency kernels.  For regularization, we add a penalty on $\|\vu(\vx')\|^2$, and further assume $\vu(\vx') \geq 0$ so that the connections from grid cells to place cells are excitatory \citep{zhang2013optogenetic,Rowland2018}. However, note that $\vu(\vx') \geq 0$ is not necessary for the emergence of hexagon grid patterns as shown in the ablation studies.

Expectations in $L_0$, $L_1$ and $L_2$ are approximated by Monte Carlo samples. $L$ is minimized by {\em Adam}~\cite{kingma2014adam} optimizer. 
See Supplementary for implementation details. 

It is worth noting that, consistent with the experimental observations, we assume individual place field $A(\vx, \vx')$  to exhibit a Gaussian shape, rather than a Mexican-hat pattern (with balanced excitatory center and inhibitory surround) as assumed in previous basis expansion models ~\citep{dordek2016extracting,sorscher2019unified} of grid cells. 

\noindent{\bf ReLU non-linearity}. We also experiment with a non-linear transformation model where a ReLU activation is added. See Supplementary for details.

\section{Experiments} \label{sect:exp}
\begin{figure*}[ht]
\centering
  \includegraphics[width=.99\linewidth]{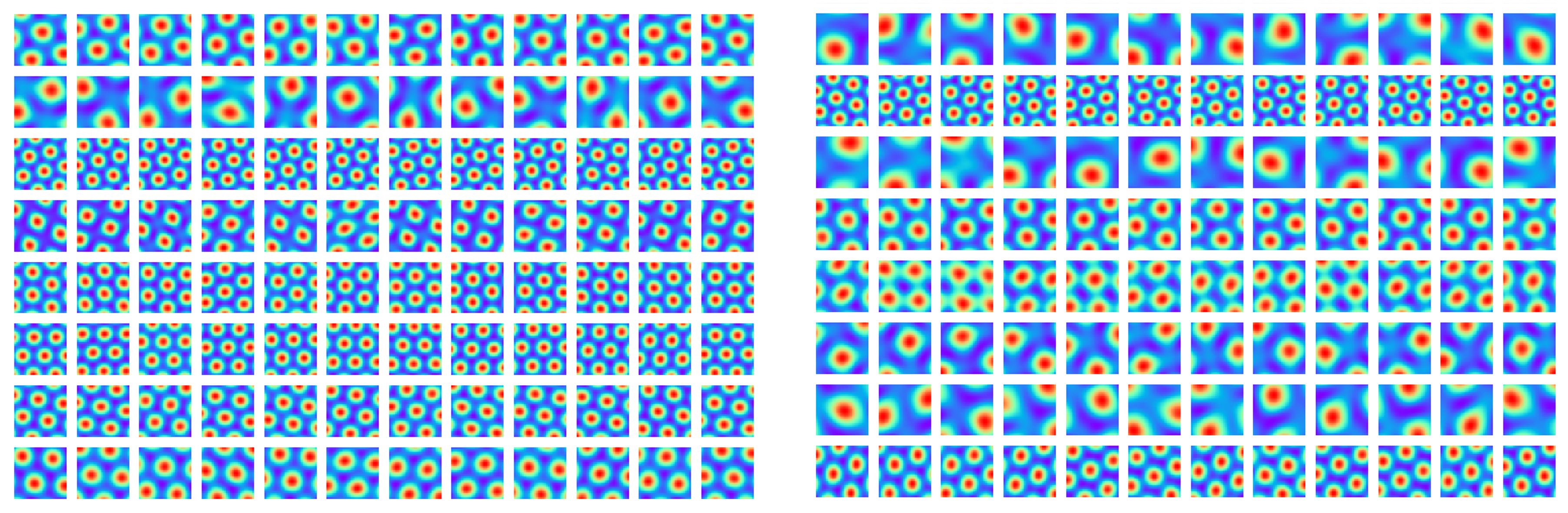} \\
  \caption{\small Hexagonal grid firing patterns emerge in the learned network. Every response map shows the firing pattern of one neuron (i.e, one element of $\vv$) in the 2D environment. Every row shows the firing patterns of the neurons within the same block or module. }
  \label{fig: units}
\end{figure*}

We conduct numerical experiments to learn the representations as described in Section \ref{sect:learning}. Specifically, we use a square environment with size $1$m $\times$ $1$m, which is discretized into a $40 \times 40$ lattice. For direction, we discretize the circle $[0, 2\pi]$ into 144 directions and use nearest neighbor linear interpolations for values in between. We use the second-order Taylor expansion (Eq. (\ref{eq:second})) to approximate the exponential map $\exp(\mB(\theta) \Delta r)$. The displacement $\Delta r$ are sampled within a small range, i.e., $\Delta r$ is smaller than $3$ grids on the lattice. For $A(\vx, \vx')$, we use a Gaussian adjacency kernel with $\sigma = 0.07$. $\vv(\vx)$ is of $d = 192$ dimensions, which is partitioned into $K=16$ modules, each of which has $12$ cells. 

\subsection{Hexagon grid patterns} 
Fig. \ref{fig: units} shows the learned firing patterns of $\vv(\vx) = (v_i(\vx), i = 1, ..., d)$  over the $40 \times 40$ lattice of $\vx$. Every row shows the learned units belonging to the same block or module.  Regular hexagon grid patterns emerge. Within each block or module, the scales and orientations are roughly the same, but with different phases or spatial shifts. For the learned $\mB(\theta)$, each element shows regular sine/cosine tuning over $\theta$. See Supplementary for more learned patterns. 

\begin{center}	\begin{minipage}[c]{.4\textwidth}
\centering
\includegraphics[width=.9\linewidth]{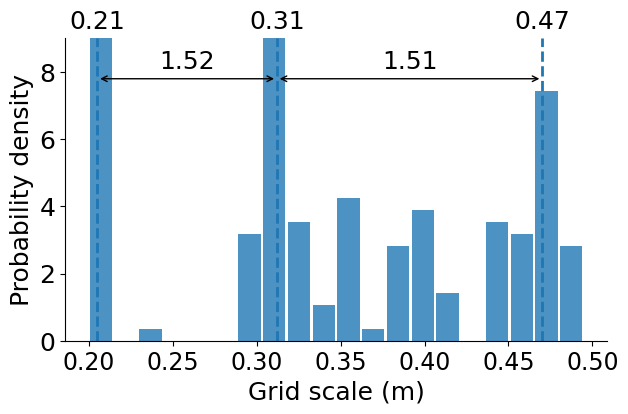}
\captionof{figure}{\small Multi-modal distribution of grid scales of the learned model grid cells. The scale ratios closely match the real data~\citep{stensola2012entorhinal}. }
\label{fig: profile} 
	\end{minipage}
	\hspace{4mm}
	\begin{minipage}[c]{.53\textwidth}
\centering
	\captionof{table}{\small Summary of gridness scores of the patterns learned from different models. To determine valid grid cells, we apply the same threshold of gridness score as in \cite{banino2018vector}, i.e., gridness score > 0.37. For our model, we run 5 trials and report the average and standard deviation.} 
	\footnotesize
 \begin{tabular}{lcc} 
    \toprule
    Model &  Gridness score ($\uparrow$) & \% of grid cells \\
    \midrule
    \cite{banino2018vector} (LSTM) & 0.18 & 25.20 \\
    \cite{sorscher2019unified} (RNN) & 0.48 & 56.10 \\ 
    Ours & {\bf 0.90} $\pm$ 0.044 & {\bf 73.10} $\pm$ 1.33 \\
    \bottomrule 
    \end{tabular}
    \label{tabl:gridness}
	\end{minipage}
\end{center}

We further investigate the characteristics of the learned firing patterns of $\vv(\vx)$ using measures adopted from the literature of grid cells. Specifically, the hexagonal regularity, scale and orientation of grid-like patterns are quantified using the gridness score, grid scale and grid orientation~\citep{langston2010development,sargolini2006conjunctive}, which are determined by taking a circular sample of the autocorrelogram of the response map. Table~\ref{tabl:gridness} summarizes the results of gridness scores and comparisons with other optimization-based approaches~\citep{banino2018vector,sorscher2019unified}. We apply the same threshold to determine whether a learned neuron can be considered a grid cell  as in \cite{banino2018vector} (i.e., gridness score > 0.37). For our model, $73.10\%$ of the learned neurons exhibit significant hexagonal periodicity in terms of the gridness score. Fig. \ref{fig: profile} shows the histogram of grid scales of the learned grid cell neurons (mean $0.33$, range $0.21$ to $0.49$), which follows a multi-modal distribution. The ratio between neighboring modes are roughly $1.52$ and $1.51$, which closely matches the theoretical predictions~\citep{Wei2015,stemmler2015connecting} and also the empirical results from rodent grid cells~\citep{stensola2012entorhinal}. Collectively, these results reveal striking, quantitative correspondence between the properties of our model neurons and those of the grid cells in the brain.

{\bf Connection to continuous attractor neural network (CANN) defined on 2D torus}. The fact that the learned response maps of each module are shifted versions of a common hexagon periodic pattern implies that the learned codebook manifold forms a 2D torus, and as the agent moves, the responses of the grid cells undergo a cyclic permutation. This is consistent with the CANN models hand-crafted on 2D torus. See Supplementary for a detailed discussion. 

{\bf Ablation studies}. We conduct ablation studies to examine whether certain model assumptions are  empirically important for the emergence of hexagon grid patterns. The conclusions are highlighted as follows: (1) The loss term $L_2$ (Eq. (\ref{eq:L3})) constraining the isotropic scaling condition is necessary for learning hexagon grid patterns. (2) The constraint $\vu(\vx') \geq 0$ is not necessary for learning hexagon patterns, but the activations can be either excitatory or inhibitory without the constraint. (3) The skew-symmetric assumption on $\mB(\theta)$ is not important for learning hexagon grid pattern. (4) Hexagon patterns always emerge regardless of the choice of block size and number of blocks. (5) Multiple blocks or modules are necessary for the emergence of hexagon grid patterns of multiple scales. See Fig. \ref{fig:ablation} for several learned patterns and Supplementary for the full studies.   

\begin{figure}[h]
	\centering	
	\begin{tabular}{ccc}
\includegraphics[width=.21\linewidth]{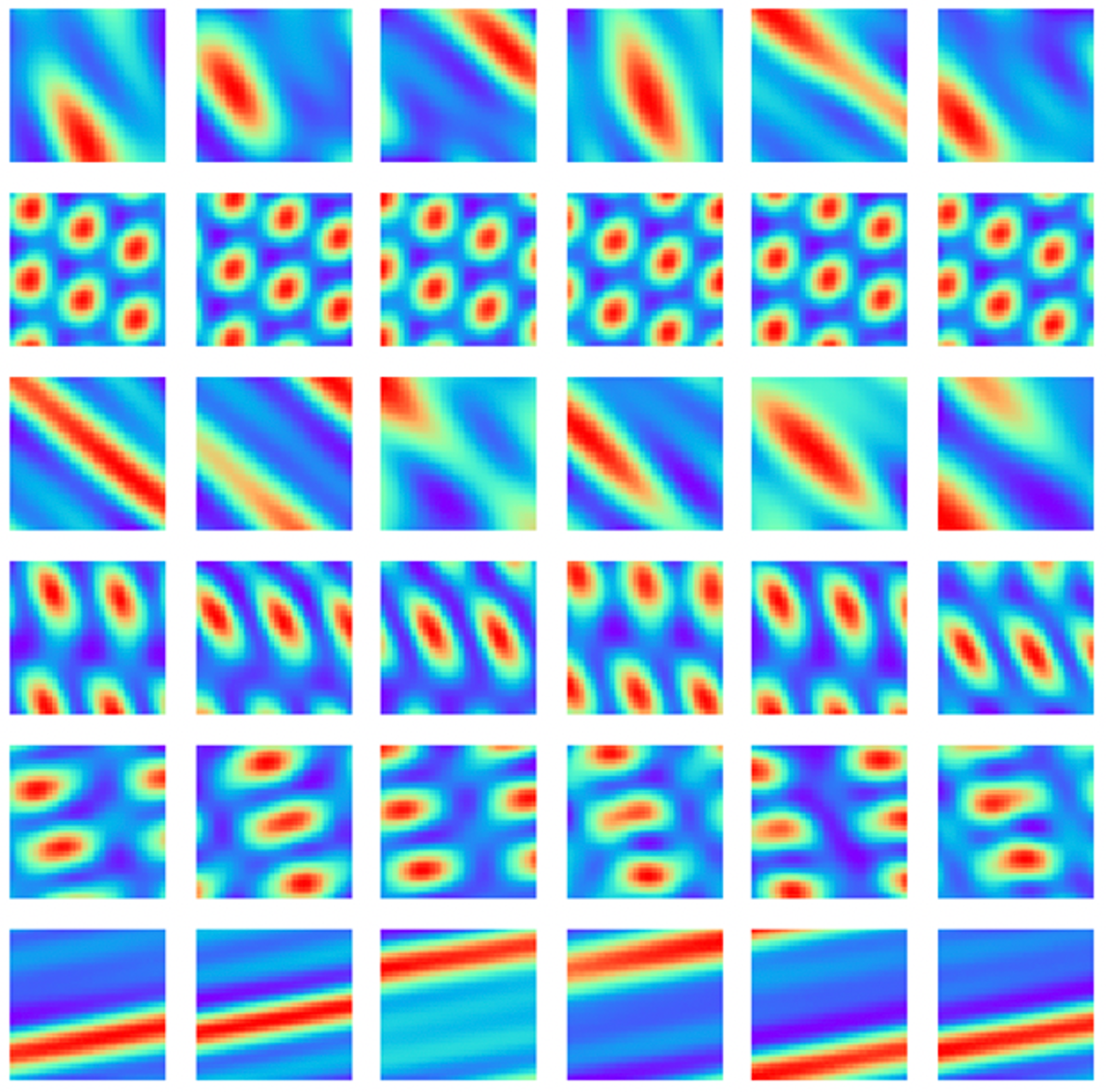}  &
		\includegraphics[width=.21\linewidth]{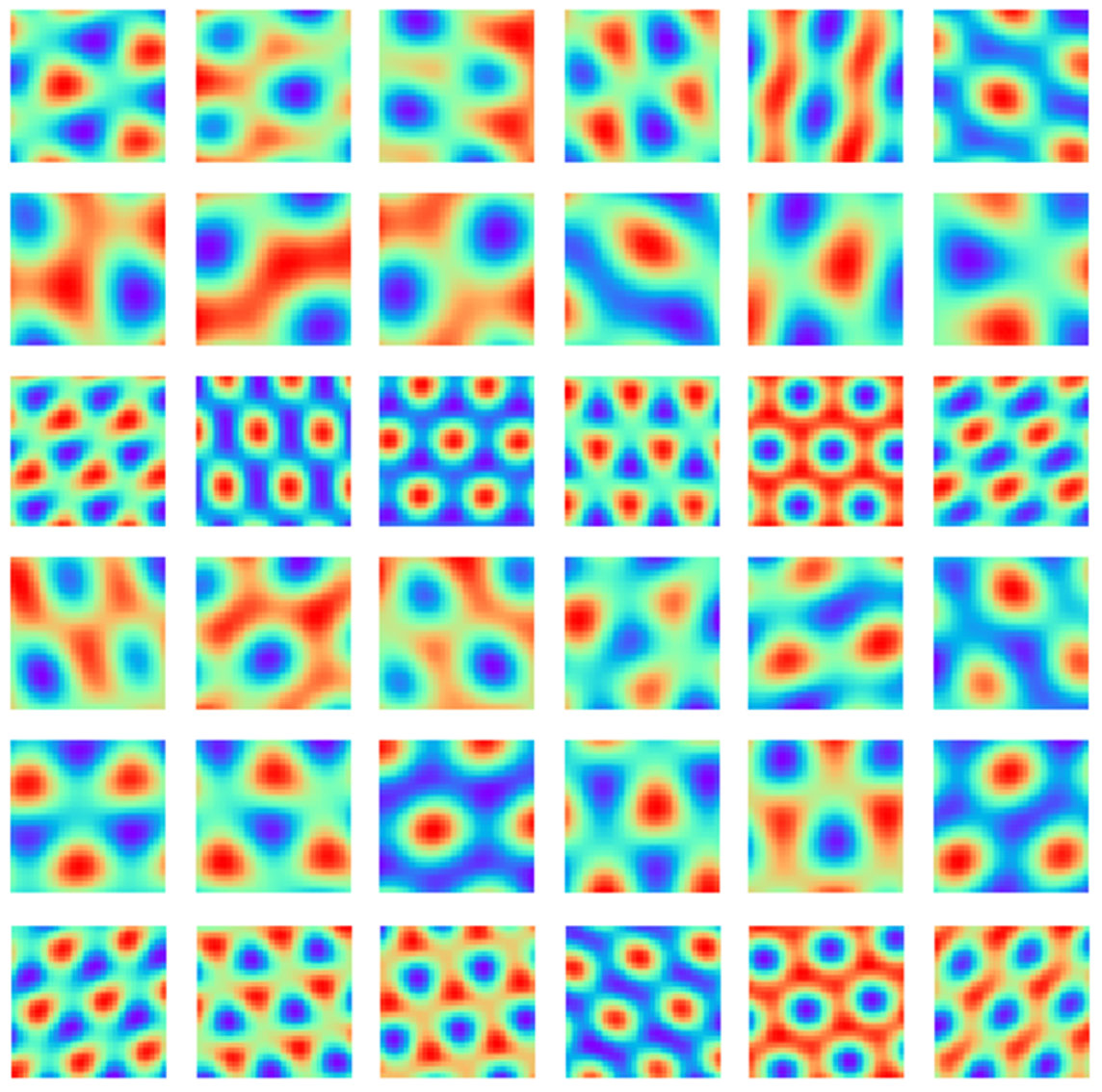} &
		\includegraphics[width=.21\linewidth]{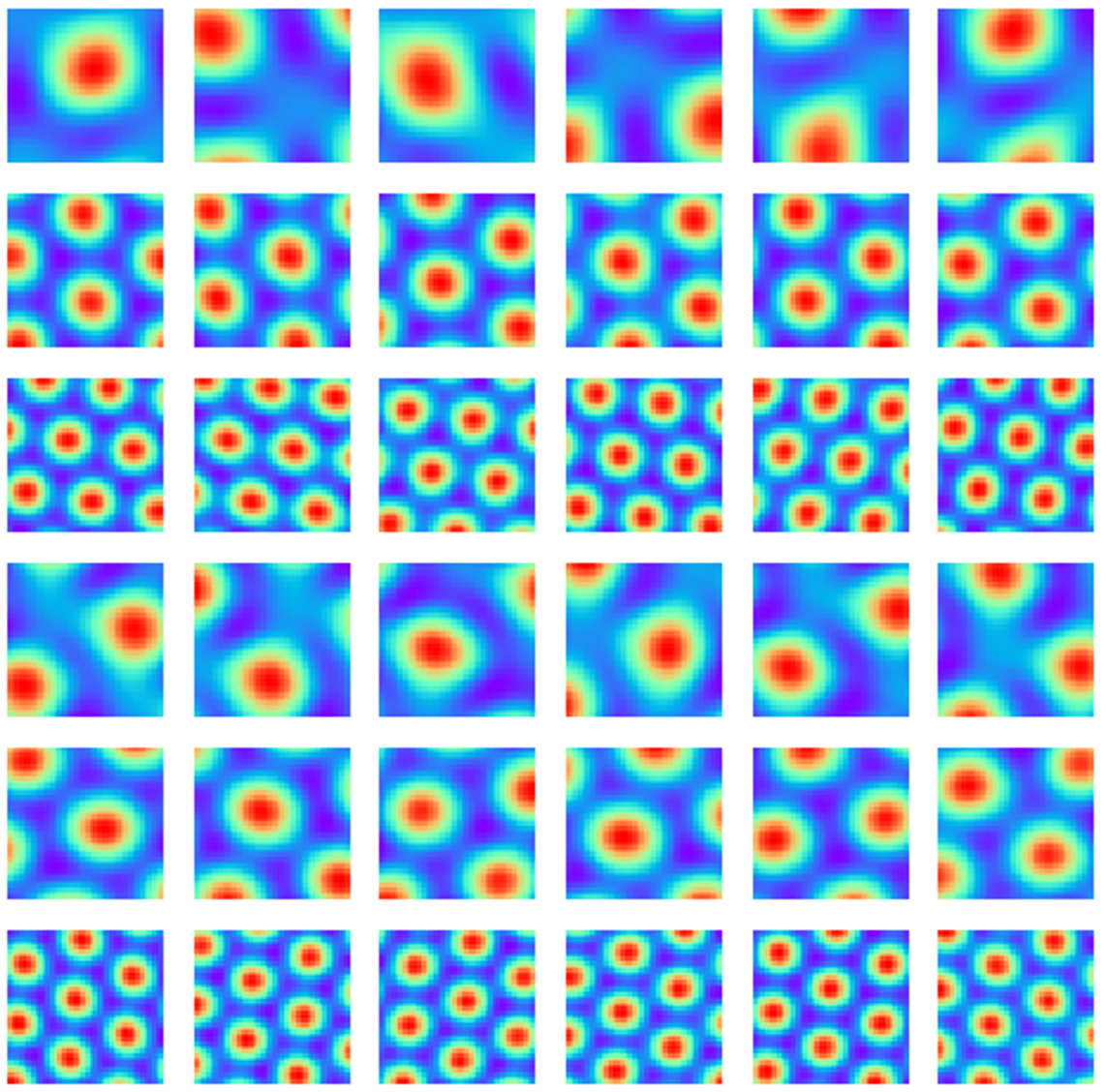}  \\
		{\small (a) without $L_2$} & {\small  (b) without $\vu(\vx') \geq 0$} & {\footnotesize (c) without skew-symmetry}  
							\end{tabular} 						
	\caption{\small Learned response maps in ablation studies where a certain model assumption is removed. (a) Remove the loss term $L_2$. (b) Remove the assumption $\vu(\vx') \geq 0$. (c) Remove the skew-symmetric assumption on $\mB(\theta)$. % See Supplementary for the full learned neurons.
	}   
	%\vspace{-.4cm}
	\label{fig:ablation}
\end{figure}

\subsection{Path integration}  
We then examine the ability of the learned model on performing multi-step path integration, which can be accomplished by recurrently updating $\vv_t$ (Eq. (\ref{eq:PI})) and decoding $\vv_t$ to $\vx_t$ for $t=1, ..., T$ (Eq. (\ref{eq:decode})). Re-encoding $\vv_t \leftarrow \vv(\vx_t)$ after decoding is adopted. Fig. \ref{fig: path}(a) shows an example trajectory of accurate path integration for number of time steps $T=30$. As shown in Fig. \ref{fig: path}(b), with re-encoding, the path integration error remains close to zero over a duration of $500$ time steps ($ < 0.01$ cm, averaged over 1,000 episodes), even if the model is trained with the single-time-step transformation model (Eq. (\ref{eq:L2})). Without re-encoding, the error goes slight higher but still remains small (ranging from $0.0$ to $4.2$ cm, mean $1.9$ cm in the 1m $\times$ 1m environment). Fig. \ref{fig: path}(c)  summarizes the path integration performance by fixing the number of blocks and altering the block size. The performance of path integration would be improved as the block size becomes larger, i.e., with more neurons in each module. When block size is larger than $16$, path integration is very accurate for the time steps tested. 

{\bf Error correction}. See Supplementary for numerical experiments on error correction, which show that the learn model is still capable of path integration when we apply Gaussian white noise errors or Bernoulli drop-out errors to $\vv_t$. 
\begin{figure}[h] 
\centering
  \includegraphics[width=.99\linewidth]{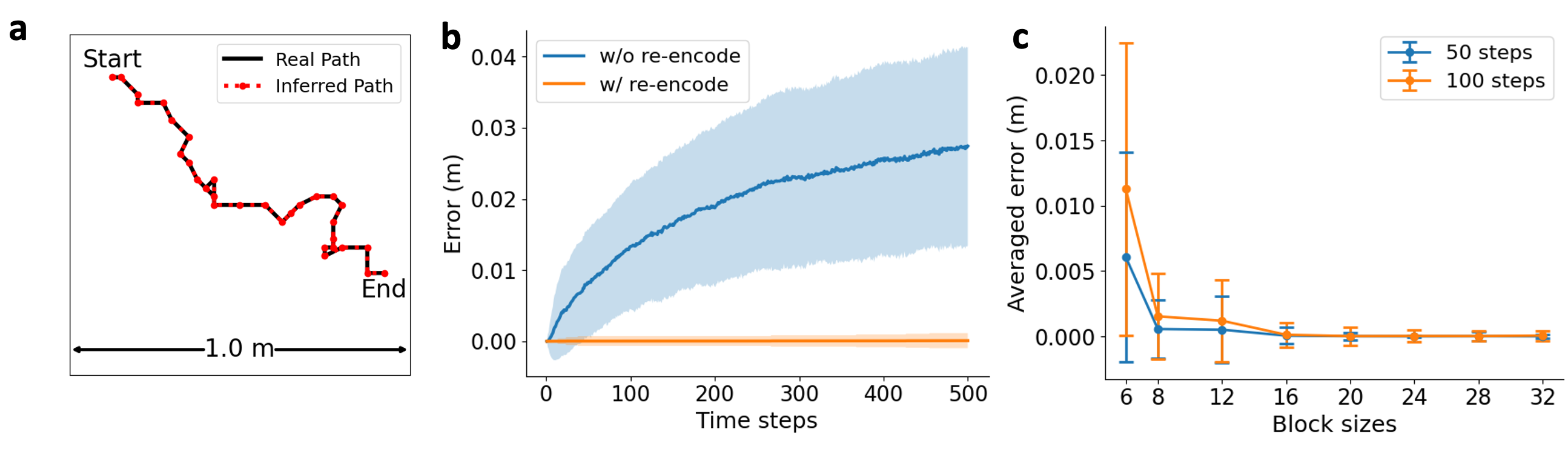}
  \caption{\small The learned model can perform accurate path integration. (a) Black: example trajectory. Red: inferred trajectory. (b) Path integration error over number of time steps, for procedures with re-encoding and without re-encoding. (c) Path integration error with fixed number of blocks and different block sizes, for 50 and 100 time steps. The error band in (b) and error bar in (c) are standard deviations computed over 1,000 episodes.  } % The decoded self-positions accurately matches the real path. 
 % \vspace{-.4cm}
  \label{fig: path}
\end{figure}

\subsection{Additional experiments on path planning and egocentric vision} 

We also conduct additional experiments on path planning and egocentric vision with our model. Path planning can be accomplished by steepest ascent on the adjacency to the target position. For egocentric vision, we learn an extra generator network that generates the visual image given the position encoding formed by the grid cells. See Supplementary for details.

\section{Related work}

Our work is related to several lines of previous research on modeling grid cells. First, RNN models have been used to model grid cells and path integration. The traditional approach uses
simulation-based models with hand-crafted connectivity, known as continuous attractor neural network (CANN)~\citep{amit1992modeling,burak2009accurate,Couey2013,Pastoll2013,Agmon2020}.
On the other hand, more recently two pioneering papers~\citep{cueva2018emergence,banino2018vector} developed optimization-based RNN approaches to learn the path integration model and discovered that grid-like response patterns can emerge in the optimized networks. These results are further substantiated in~\citep{sorscher2019unified,Cueva2020}.
Our work analyzes the properties of the general recurrent model for path integration, and these properties seem to be satisfied by the hand-crafted CANN models.  Our method belongs to the scheme of optimization-based approaches, and the learned response maps share similar properties as assumed by the CANN models.

Second,  our work differs from the PCA-based basis expansion models~\citep{dordek2016extracting,sorscher2019unified,stachenfeld2017hippocampus} in that, unlike PCA, we make
no assumption about the orthogonality between the basis functions, and the basis functions are generated by the transformation model.  
Furthermore, in previous basis expansion models ~\citep{dordek2016extracting,sorscher2019unified}, place fields with Mexican-hat patterns (with balanced excitatory center and inhibitory surround) had to be assumed in order to obtain hexagonal grid firing patterns. However, experimentally measured place fields in biological brains were instead well characterized by Gaussian functions. Crucially, in our model, hexagonal grids emerge from learning with Gaussian place fields, and there is no need to assume any additional surround mechanisms or difference of Gaussians kernels. 

In another related paper, \citep{gao2018learning} proposed matrix representation of 2D self-motion, while our work analyzes general transformations. Our investigation of the special case of linear transformation model reveals the matrix Lie group and the matrix Lie algebra of rotation group. Our work also connects the linear transformation model to the basis expansion model via unitary group representation theory.  

\section{Conclusion} \label{sect:conclusion} 

This paper analyzes the recurrent model for path integration calculations by grid cells. We identify a group representation condition and an isotropic scaling condition that give rise to locally conformal embedding of the self-motion. We study a  linear prototype model that reveals the matrix Lie group of rotation, and explore the connection between the isotropic scaling condition and hexagon grid patterns. In addition to these theoretical investigations, our numerical experiments demonstrate that our model can learn hexagon grid patterns for the response maps of grid cells, and the learned model is capable of accurate long distance path integration. 

In this work, the numerical experiments are mostly limited to the linear transformation model, with the exception of an experiment with ReLU non-linearity. We will conduct experiments on the other non-linear transformation models, especially the forms assumed by the hand-crafted continuous attractor neural networks. Moreover, we assume that the agent navigates within a square open-field environment without obstacles or rewards. It is worthwhile to explore more complicated environments, including 3D environment.

\begin{ack}
The work was supported by NSF DMS-2015577, ONR MURI project N00014-16-1-2007, DARPA XAI project N66001-17-2-4029, and XSEDE grant ASC170063. We thank Yaxuan Zhu from UCLA Department of Statistics for his help with experiments on egocentric vision.  We thank Dr. Wenhao Zhang for sharing his knowledge and insights on continuous attractor neural networks.  We thank Sirui Xie for discussions. We thank the three reviewers for their constructive comments. 
\end{ack}

\bibliographystyle{ieee_fullname}
\bibliography{neurips_2021}

%%%%%%%%%%%%%%%%%%%%%%%%%%%%%%%%%%%%%%%%%%%%%%%%%%%%%%%%%%%%

\newpage
\section*{Supplementary Materials}
\appendix

\section{Theoretical analysis}

\subsection{Graphical illustrations of key equations} 

Fig. \ref{fig:g1} illustrates key equations in the main text as well as in the supplementary materials. 

\begin{figure}[h]	
	\centering	
	\begin{tabular}{cc}
\includegraphics[width=.25\linewidth]{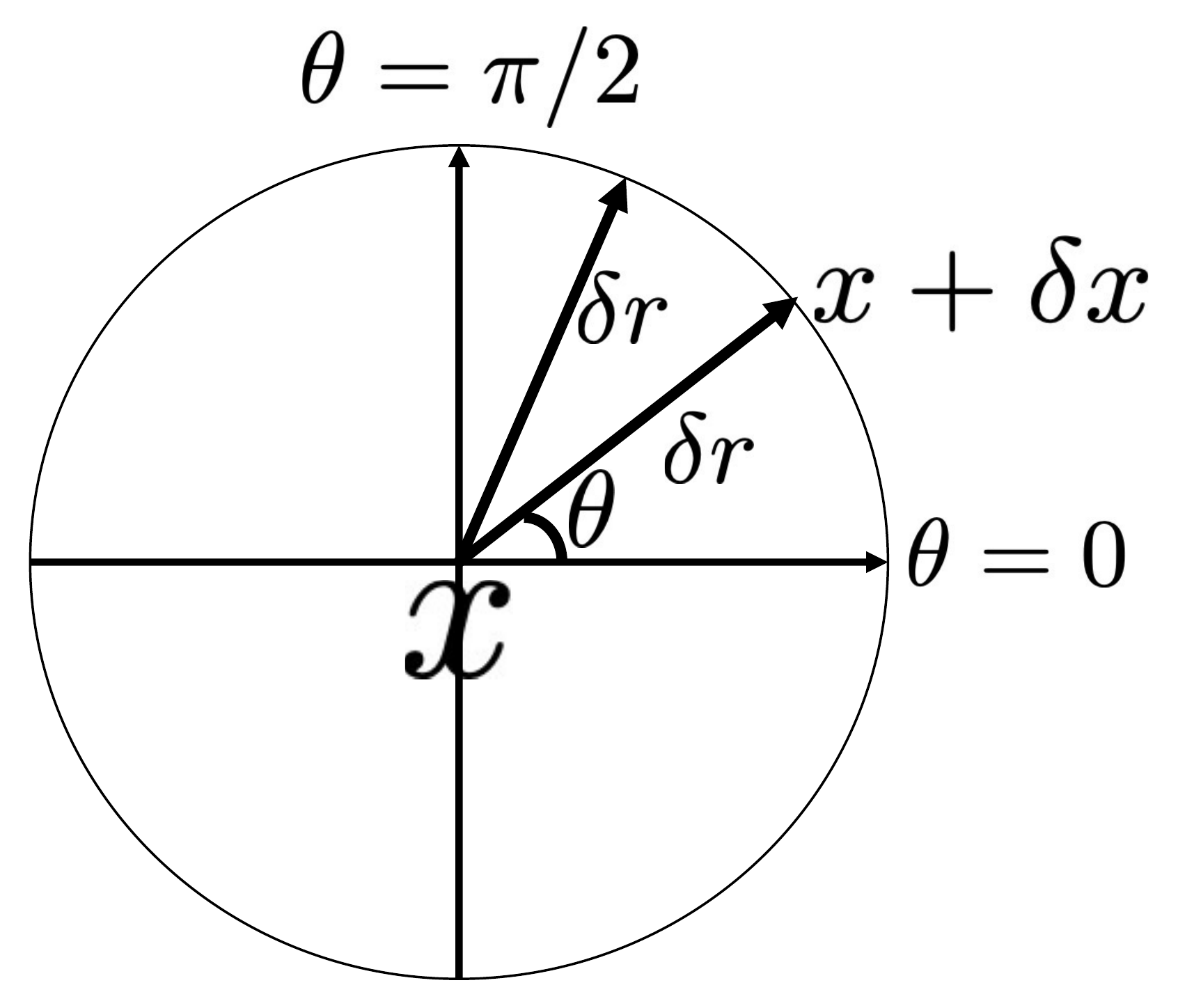}  &
		\includegraphics[width=.3\linewidth]{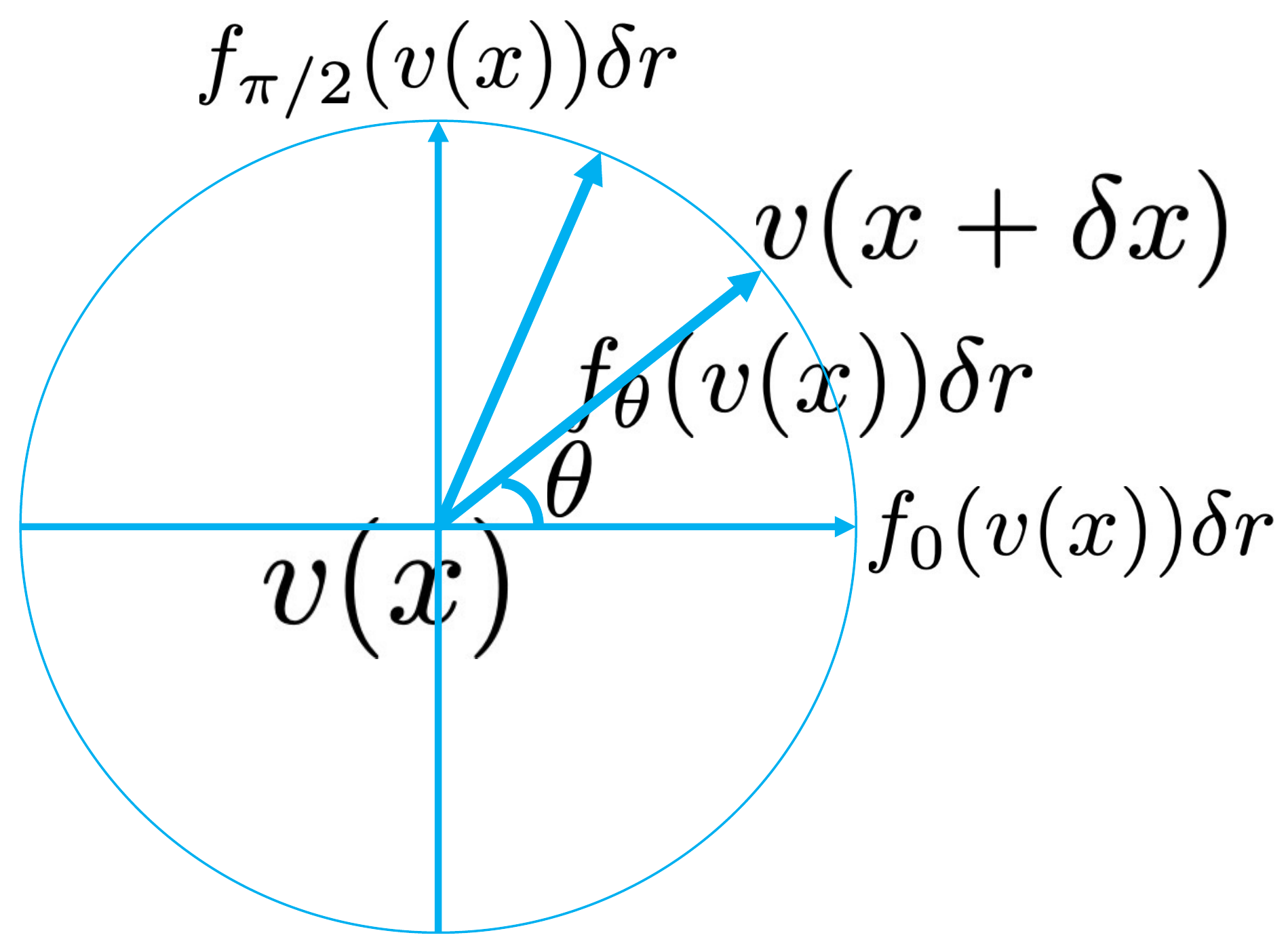}  \\
				(a) physical space & (b) neural space \\
		\includegraphics[width=.3\linewidth]{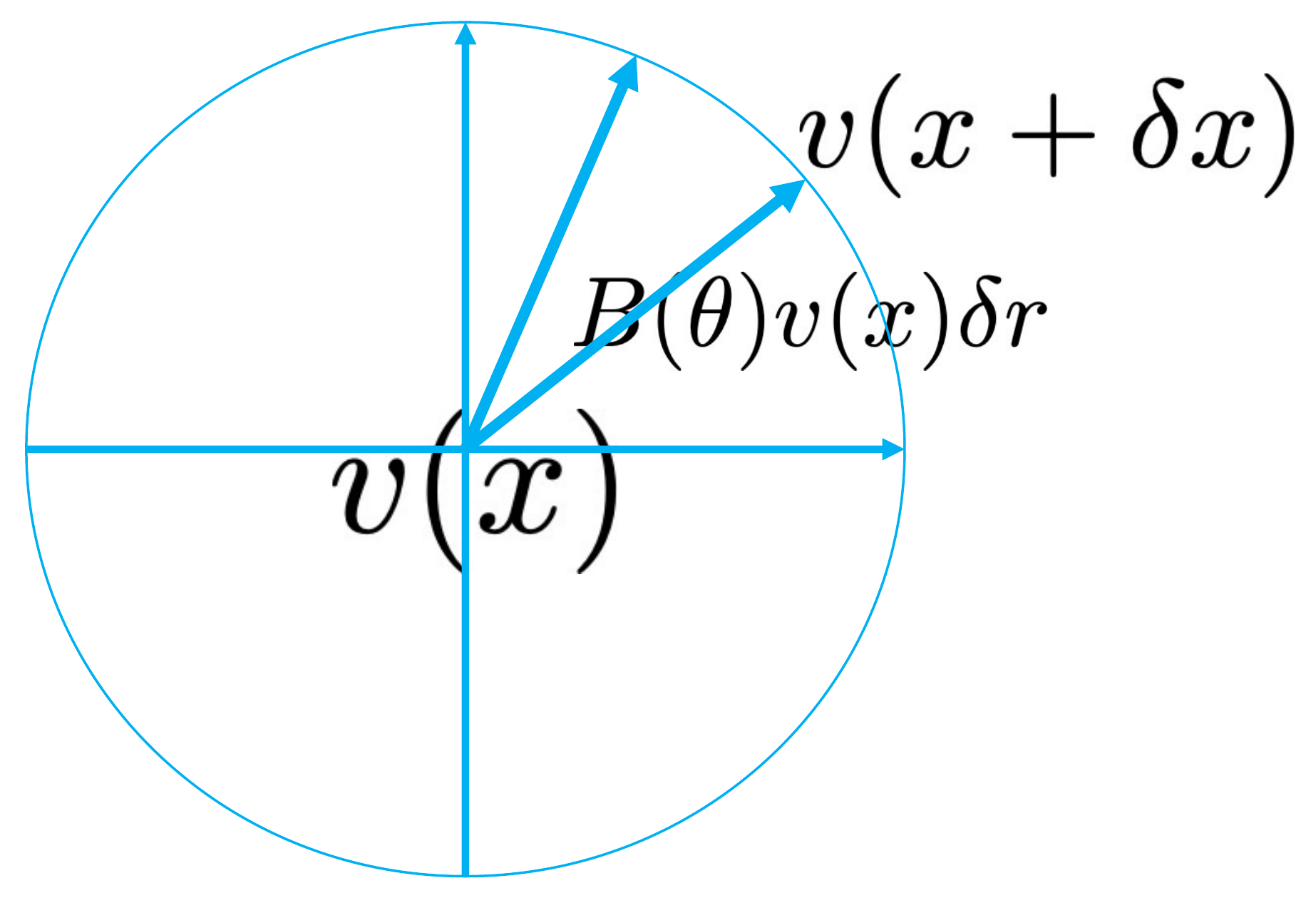}  &
		\includegraphics[width=.3\linewidth]{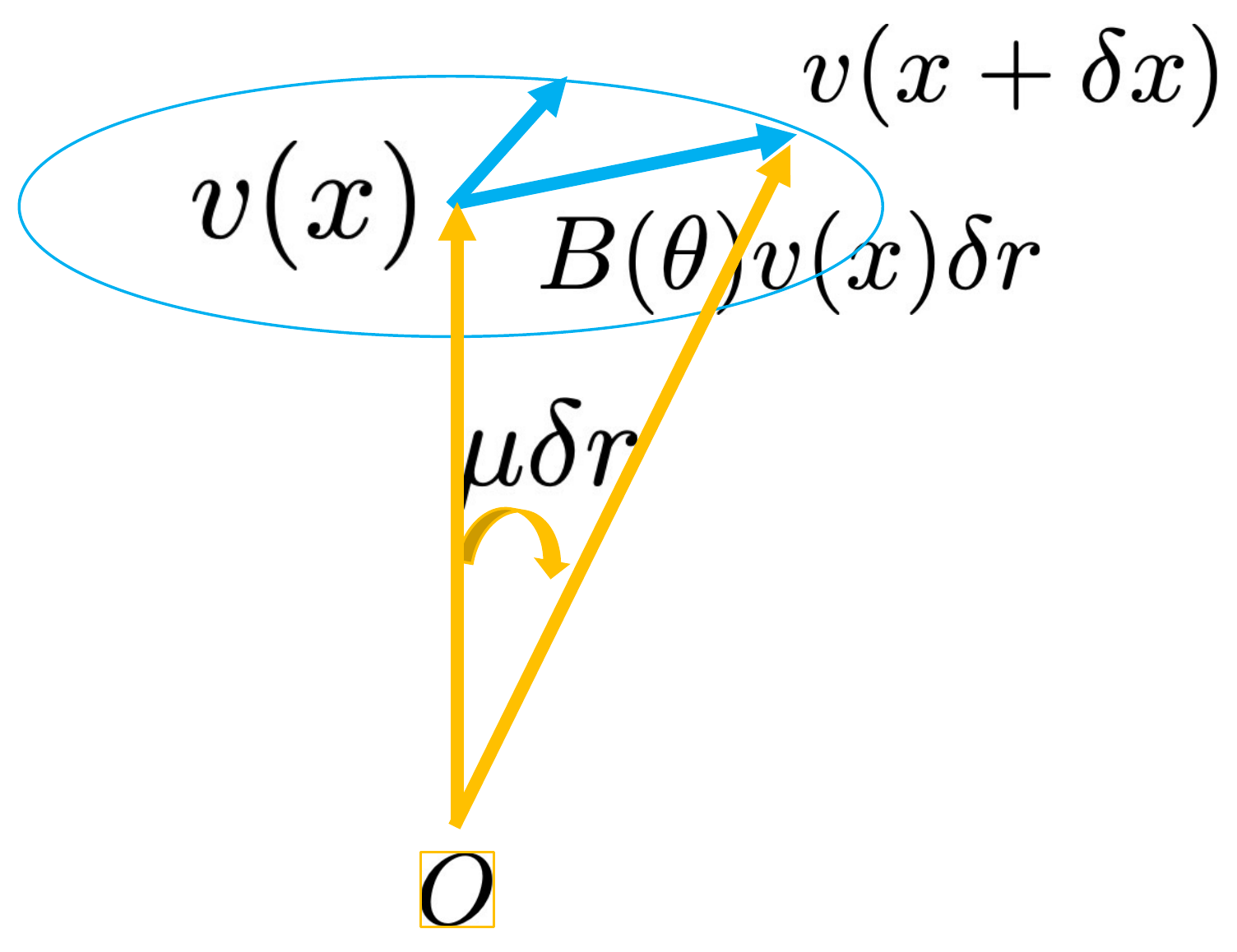} \\
 (c) linear transformation & (d) as rotation
							\end{tabular} 						
	\caption{\small Color-coded illustration. (a) In the 2D physical space, the agent moves from $\vx$ to $\vx+\delta \vx$, where $\delta \vx = (\delta r \cos \theta, \delta r \sin \theta)$, i.e., the agent moves by $\delta r$ along the direction $\theta$. We also show a displacement of $\delta r$ in a different direction. (b) In the $d$-dimensional neural space, the vector $\vv(\vx)$ is changed to $\vv(\vx+\delta \vx) = F(\vv(\vx), \delta r, \theta) = \vv(\vx) + f_\theta(\vv(\vx)) \delta r + o(\delta r)$, where the displacement is $f_\theta(\vv(\vx)) \delta r = f_{0}(\vv(\vx)) \delta r \cos \theta + f_{\pi/2}(\vv(\vx))\delta r \sin \theta$. Under the isotropic condition that $\|f_\theta(\vv(\vx))\|$ is constant over $\theta$, the local 2D self-motion $\delta \vx$ at $\vx$ in the 2D physical space is embedded conformally into the neural space as a 2D subspace around $\vv(\vx)$. (c) Linear transformation, where $f_\theta(\vv(\vx)) = \mB(\theta) \vv(\vx)$. (d) 3D perspective view of linear transformation as a rotation: $\vv(\vx+\delta \vx)$ is a rotation of $\vv(\vx)$, and the angle of rotation is $\mu \delta r$, where $\mu = \|\mB(\theta) \vv(\vx)\|/\|\vv(\vx))\|$ ($\mu$ may depend on $\vx$). }
	\label{fig:g1}
\end{figure}

\subsection{Proof of Theorem 1 on conformal embedding}

{\em Proof:} See Fig. \ref{fig:g1}(a) and (b) for an illustration. Consider the self-motion $\delta \vx = (\delta r \cos \theta, \delta r \sin \theta)$, 
 \begin{align} 
  \vv(\vx + \delta \vx) &= F(\vv(\vx), \delta r, \theta) = \vv(\vx) + f_\theta(\vv(\vx)) \delta r + o(\delta r). \label{eq:general}
  \end{align}
We can decompose the self-motion $\delta \vx$ into two steps. First move along the direction $0$ by $\delta r \cos \theta$, and then move along the direction $\pi/2$ by $\delta r \sin \theta$. Then under the {\bf group representation condition}: 
 \begin{align} 
    \vv(\vx + \delta \vx) 
    &= F[F(\vv(\vx), \delta r \cos \theta, 0),  \delta r \sin \theta, \pi/2)] \nonumber \\
   & = F[\vv(\vx) + f_0(\vv(\vx)) \delta r \cos \theta + o(\delta r), \delta r \sin \theta, \pi/2] \nonumber \\
   &=[\vv(\vx) + f_0(\vv(\vx)) \delta r \cos \theta] 
   + f_{\pi/2}[\vv(\vx) + f_0(\vv(\vx)) \delta r \cos \theta + o(\delta r)] \delta r \sin \theta + o(\delta r) \nonumber \\
    &= \vv(\vx) + f_0(\vv(\vx)) \delta r \cos \theta + f_{\pi/2}(\vv(\vx)) \delta r \sin \theta + o(\delta r), 
    \label{eq:group}
 \end{align}
 The last equation holds because assuming the derivative $f'_{\pi/2}(\vv(\vx))$ exists, then by first-order Taylor expansion,
 \begin{align}
 & f_{\pi/2}[\vv(\vx) + f_0(\vv(\vx)) \delta r \cos \theta + o(\delta r)] \delta r \sin \theta \\
  = & [f_{\pi/2}(\vv(\vx)) + f'_{\pi/2}(\vv(\vx)) f_0(\vv(\vx))\delta r \cos\theta + o(\delta r)] \delta r \sin \theta \\
  = & f_{\pi/2}(\vv(\vx)) \delta r \sin \theta +  o(\delta r).
 \end{align} 

 Since $\vv(\vx + \delta \vx) = \vv(\vx) + f_\theta(\vv(\vx)) \delta r + o(\delta r)$, by Eq. (\ref{eq:group}) we have $f_\theta(\vv(\vx)) = f_0(\vv(\vx)) \cos \theta + f_{\pi/2}(\vv(\vx)) \sin \theta$, which is a 2D basis expansion. We are yet to prove that the two basis vectors $f_0(\vv(\vx))$ and $f_{\pi/2}(\vv(\vx))$ are orthogonal with equal norm. 
 
 For notational simplicity, let $\vv_1 = f_0(\vv(\vx))$ and $\vv_2 = f_{\pi/2}(\vv(\vx))$. Then under the {\bf isotropic scaling condition}, $\|\vv_1\| = \|\vv_2\| = \|f_\theta(\vv(\vx))\| = s$, and $f_\theta(\vv(\vx)) = \vv_1 \cos \theta + \vv_2 \sin \theta$ for any $\theta$. Then we have that for any $\theta$, 
 \begin{align} 
s^2 &=  \|f_\theta(\vv(\vx))\|^2 = \|\vv_1 \cos \theta + \vv_2 \sin \theta\|^2 
  = s^2  + 2 \langle \vv_1, \vv_2 \rangle \cos \theta \sin \theta. 
\end{align}
Thus $\langle \vv_1, \vv_2\rangle = 0$, i.e., $f_0(\vv(\vx)) \perp f_{\pi/2}(\vv(\vx))$. This leads to the conformal embedding of the local 2D polar system in the physical space as a 2D polar system in the $d$-dimensional neural space, with a scaling factor $s$ (which may depend on $\vx$). $\square$

\subsection{Proofs of Theorem 2 and Proposition 1 on error correction} \label{supp:err}

{\em Proof of Theorem 2:} By Theorem 1, for a fixed self-position $\vx$, we embed the 2D local neighborhood around $\vx$ as a local 2D plane around $\vv(\vx)$ in the $d$-dimensional neural space. A local perturbation in self-position, $\delta \vx$, is translated into a local perturbation in $\vv(\vv+\delta \vx)$, so that 
\begin{align}
  \|\delta \vv\|^2 = \|f_\theta(\vv(\vx)) \delta r + o(\delta r)\|^2 =  s^2 \|\delta \vx\|^2, 
\end{align}
where $\delta \vv = \vv(\vx+\delta \vx) - \vv(\vx)$.

Suppose the agent infers its 2D position $\hat{\vx}$ by $\hat{\vx} = \arg\min_{\vx'} \|\vv - \vv(\vx')\|^2$, which amounts to projecting $\vv$  onto the local 2D plane around $\vv(\vx)$. The projected vector $\vv(\hat{\vx})$ on the local 2D plane is $\vv(\vx) + \delta \vv$,  where $\delta \vv$ is the projection of ${\bf \epsilon}$ onto the 2D plane. More specifically, let $(\vv_1, \vv_2)$ be an orthonormal basis of the local 2D plane centered at $\vv(\vx)$. Then $\delta \vv$ can be written as $e_1  \vv_1 + e_2  \vv_2$,  where  
\begin{align}
\ve = (e_1, e_2)^\top = (\vv_1, \vv_2)^\top \epsilon \sim \mathcal{N}(0, \tau^2 \mI_2).
\end{align} 
Let $\delta \vx = \hat{\vx} - \vx$. 
Due to {\bf isotropic scaling and conformal embedding}, the $\ell_2$ squared error  translate according to 
\begin{align}
\|\delta \vx\|^2 = \|\delta \vv\|^2/s^2 = (e_1^2 + e_2^2)/s^2,
\end{align}
whose expectation is $2 \tau^2/s^2$. Thus $\E\|\hat{\vx} - \vx\|^2 = 2 \tau^2/s^2$.$\square$

{\em Proof of Proposition 1:} It is reasonable to assume $\tau^2 = \alpha^2 (\|\vv(\vx)\|^2/d)$,  where $\alpha^2$ measures the variance of noise relative to $\|\vv(\vx)\|^2/d$, which is the average of $(v_i(\vx)^2, i = 1, ..., d)$. In other words, $\alpha^2$ measures the noise level. 

In the linear case, the metric is 
\begin{align}
\mu = \|f_\theta(\vv(\vx))\|/\|\vv(\vx)\| = \|\mB(\theta) \vv(\vx))\|/\|\vv(\vx)\| = s/\|\vv(\vx)\|,
\end{align}
 which measures how fast $\vv(\vx)$ rotates in the neural space as $\vx$ changes. Then 
\begin{align}
 \E  \|\delta \vx\|^2 = 2 \alpha^2 /(\mu^2 d). 
\end{align}
The above scaling shows that error correction depends on two factors. One is the metric $\mu$, and the other is the dimensionality $d$, i.e., the number of neurons. These correspond to two phases of error correction. One is to project the $d$-dimensional ${\bf \epsilon}$ to the 2-dimensional $\delta \vv$. The bigger $d$ is, the better the error correction. The other is to translate $\|\delta \vv\|^2$ to $\|\delta \vx\|^2$. The bigger $\mu$ is, the better the error correction. $\square$

\subsection{Proof of Theorem 4 on hexagon grid patterns}

 {\em Proof:} Let $\ve(\vx) = (\exp(i \langle \va_j,\vx\rangle), j = 1, 2, 3)^\top$, where $(\va_j, j = 1, 2, 3)$ are three 2D vectors of equal norm, and the angle between every pair of them is  $2\pi/3$. Let $\vv(\vx) = \mU \ve(\vx)$, where $\mU$ is an arbitrary unitary matrix, i.e., $\mU^*\mU = \mI$. Then $\|\vv(\vx)\|^2 = \|\ve(\vx)\|^2 = 3$, $\forall \vx$, and  $\ve(\vx) = \mU^* \vv(\vx)$. For self-motion $\delta \vx = (\delta r \cos \theta, \delta r \sin \theta) =  \vq(\theta) \delta r$, let 
 \begin{align}
 \Lambda(\delta \vx, \theta) &= {\rm diag}(\exp(\langle \va_j, \delta \vx\rangle), j = 1, 2, 3) \nonumber \\
 &=  {\rm diag}(\exp(\langle \va_j, \vq(\theta)\rangle \delta r), j = 1, 2, 3) \nonumber\\
 &= \mI + {\rm diag}(i \langle \va_j, \vq(\theta)\rangle), j = 1, 2, 3) \delta r + o(\delta r) \nonumber\\
 &= \mI + \mD(\theta) \delta r + o(\delta r). 
 \end{align}
 Then 
 \begin{align} 
    \vv(\vx+\delta \vx) &= \mU \ve(\vx+\delta x)\nonumber \\
    &= \mU \Lambda(\delta \vx, \theta) \ve(\vx) \nonumber\\
    &= \mU \Lambda(\delta \vx, \theta) \mU^* \vv(\vx) \nonumber\\
    &=( \mI + \mU \mD(\theta) \mU^* \vv(\vx) \delta r)\vv(\vx) + o(\delta r) \nonumber\\
    &=( \mI + \mB(\theta) \delta r )\vv(\vx) + o(\delta r),
 \end{align} 
 where $\mB(\theta) = \mU \mD(\theta) \mU^*$, and $\mB(\theta) = - \mB(\theta)^*$. 
 For isotropic condition, 
 \begin{align}
    \|\mB(\theta) \vv(\vx)\|^2 &= \| \mD(\theta) \ve(\vx)\|^2 \nonumber\\
    &= \sum_{j=1}^3 \langle \va_j, \vq(\theta)\rangle^2 \nonumber\\
    &= {\rm const}\|\va_j\|^2 \|\vq(\theta)\|^2 = {\rm const} \|\va_j\|^2, 
 \end{align}
 which is independent of $\theta$, because $(\va_j, j = 1, 2, 3)$ forms a {\bf tight frame} in 2D. 
 
 One example of $\mU$ is the following matrix: 
\begin{eqnarray}
\frac{1}{\sqrt{3}}
\begin{pmatrix}
1 & 1 & 1\\
1 & \exp(i 2 \pi/3) & \exp(- i 2 \pi/3) \\
1 & \exp(- i 2 \pi/3) & \exp( i 2 \pi/3) 
\end{pmatrix}
\end{eqnarray}
The resulting $(v_i(x), i = 1, 2, 3)$ have the same orientation but different phases, i.e., they are spatially shifted versions of each other. $\square$

The limitation of Theorem 4 is that we only show $\vv(\vx) = \mU \ve(\vx)$ satisfies the linear model and the isotropic scaling condition, but we did not show that linear model with isotropic condition only has solutions that are hexagon grid patterns. 

\subsection{From group representation to orthogonal basis functions} \label{supp:orth}

Group representation is a central theme in modern mathematics and physics. In particular, it leads to a deep understanding and   generalization of Fourier analysis or harmonic analysis. 

For the set of $(\Delta \vx)$ that form a group, a matrix representation $\mM(\Delta \vx)$ is equivalent to another representation $\tilde{\mM}(\Delta \vx)$ if there exists an invertible matrix $\mP$ such that $\tilde{\mM}(\Delta \vx) = \mP \mM(\Delta \vx) \mP^{-1}$ for each $\vx$. A matrix representation is reducible if it is equivalent to a block diagonal matrix representation, i.e., we can find a matrix $\mP$, such that $\mP \mM(\Delta \vx) \mP^{-1}$ is block diagonal for every $\Delta \vx$. Suppose the group is a finite group or a compact Lie group, and $\mM$ is a unitary representation, i.e., $\mM(\Delta \vx)$ is a unitary matrix. If $\mM$ is block-diagonal, $\mM = {\rm diag}(\mM_{k}, k = 1, ..., K)$, with non-equivalent blocks, and each block $\mM_{k}$ cannot be further reduced, then the matrix elements $(M_{kij}(\Delta \vx))$ are orthogonal basis functions of $\Delta \vx$. Such orthogonality relations are proved by Schur~\citep{zee2016group} for finite group, and by Peter-Weyl for compact Lie group~\citep{taylor2002lectures}. For our case, theoretically the group of displacements $\Delta \vx$ in the 2D domain is $\R^2$, but we learn our model within a finite range, and we further discretize the range into a lattice. Thus the above orthogonal relations hold. 

 In our model, we also assume block diagonal $\mM$, and we call each block a module. However, we do not assume each module is irreducible, i.e., each module itself may be further diagonalized into a block diagonal matrix of irreducible sub-blocks. Thus the elements within the same module $\vv_{k}(\vx)$ may be linear mixings of orthogonal basis functions of the irreducible sub-blocks, and the linear mixings themselves are not necessarily orthogonal. 

 \begin{figure}[ht]
\centering
  \includegraphics[width=.5\linewidth]{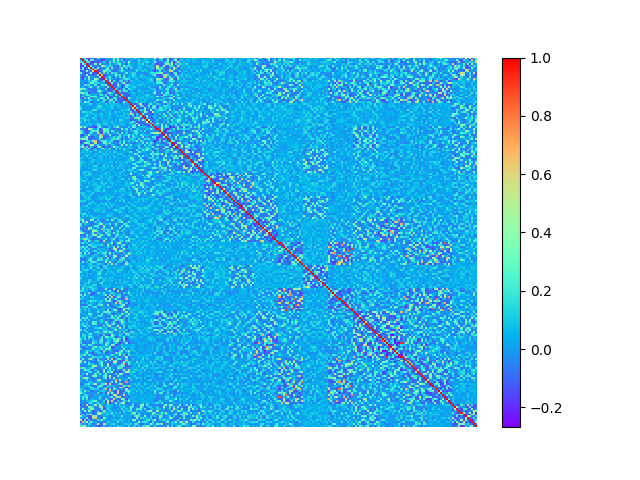}
  \caption{\small Correlation heatmap for each pair of the learned $v_i(\vx)$ and $v_j(\vx)$. The correlations are computed over $40 \times 40$ lattice of $\vx$.}
  \label{fig: correlation}
\end{figure} 

Fig. \ref{fig: correlation} visualizes the correlation between pairs of the learned $\vv_i(\vx)$ and $\vv_j(\vx)$, $i,j=1,...,d$. For different $i$ and $j$, the correlations between different $\vv_i(\vx)$ and $\vv_j(\vx)$ are close to zero; i.e., they are nearly orthogonal to each other. The average absolute value of correlation is 0.09, and the within-block average value is about the same as the between-block average value. 

Unlike previous work on learning basis expansion model (or PCA-based model \citep{dordek2016extracting}), we {\bf do not constrain the basis functions $\vv(\vx) = (\evv_i(\vx), i= 1, ..., d)$ to be orthogonal to each other}. Instead, we constrain them by our path integration model  via the loss term $L_1$. Nonetheless, the learned $\evv_i(\vx)$ are close to being orthogonal in our experiments.

\subsection{Decoding and re-encoding} 

In the above analysis, the projection of $\vv$ onto the local 2D plane around $\vv(\vx)$ is $\hat{\vx} = \arg\min_{\vx'}\|\vv - \vv(\vx')\|^2$, which, for the linear model, amounts to decoding $\vv$ to $\hat{\vx}$ via 
\begin{align} 
   \hat{\vx} = \arg \max_{x'} \langle \vv, \vv(\vx')\rangle, \label{eq:de1}
\end{align} 
because $\|\vv(\vx')\|^2$ is constant. We project $\vv$ to $\vv(\hat{\vx})$, which is an re-encoding of $\vv$. 

We can also perform decoding via the learned $\vu(\vx')$: 
\begin{align} 
   \hat{\vx} = \arg \max_{x'} \langle \vv, \vu(\vx')\rangle, \label{eq:de2} 
\end{align} 
and re-encoding $\vv \leftarrow \vv(\hat{\vx})$. For the above decoding, the heat map 
\begin{align}
\vh(\vx') = \langle \vv, \vu(\vx')\rangle = \langle \vv(\vx), \vu(\vx')\rangle + \langle {\bf \epsilon}, \vu(\vx')\rangle
 = A(\vx, \vx') + \ve(\vx'),
\end{align}
 where $\ve(\vx')= \langle {\bf \epsilon}, \vu(\vx')\rangle \sim \mathcal{N}(0, \alpha^2 \|\vv(\vx)\|^2 \|\vu(\vx')\|^2/d)$. For $A(\vx, \vx') = \exp(-\|\vx-\vx'\|^2/(2 \sigma^2)) = \langle \vv(\vx), \vu(\vx')\rangle$, if $\sigma^2$ is small, $A(\vx, \vx')$ decreases to 0 quickly, i.e., if $\|\vx'-\vx\| >  c$, then $A(\vx, \vx') < \exp(-c^2/(2\sigma^2))$, and the chance for the maximum of $\vh(\vx')$ to be achieved at an $\vx'$ so that $\|\vx'-\vx\| > c$ can be very small. The above analysis also provides a justification for regularizing $\|\vu(\vx')\|^2$ in learning. 

For error correction, we want to use small $\sigma^2$. However, for path planning, we need large $\sigma^2$ so that we can assess the adjacency as well as the change of the adjacency between the position on the path and the target position even if they are far apart.

 \begin{figure}[h]
\centering
  \includegraphics[width=.5\linewidth]{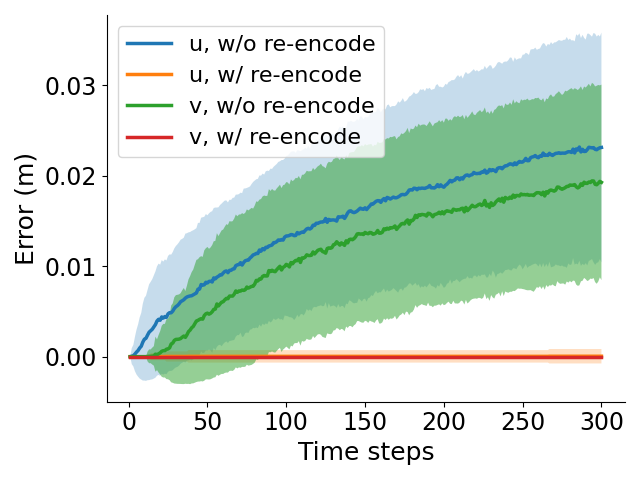}
  \caption{\small Path integration error over number of time steps. The mean and standard deviation band is computed over 1,000 episodes. ``$\vv$'' means decoding by Eq. (\ref{eq:de1}), and ``$\vu$'' means decoding by Eq. (\ref{eq:de2}). The squared domain is $1$m $\times$ $1$m.}
  \label{fig: path-integration}
\end{figure} 

In the experiments in the main text, we use Eq. (\ref{eq:de2}) for decoding. In Fig. \ref{fig: path-integration}, we also show the results of path integration using Eq. (\ref{eq:de1}) for decoding, whose performance is even better than Eq. (\ref{eq:de2}). Especially the error would remain 0 over 300 time steps and 1,000 episodes using Eq. (\ref{eq:de1}) with re-encoding. The advantage of (\ref{eq:de1}) is that error correction is achieved within the grid cells system itself without interacting with the place cells. 

\subsection{Connection to continuous attractor neural network (CANN) defined on 2D torus}

The CANN models \citep{amit1992modeling,burak2009accurate,Couey2013,Pastoll2013,Agmon2020} assume that the grid cells $\vv(\vx) = (\evv_i(\vx),  i = 1, ..., d)$ are placed on a finite 2D square lattice with periodic boundary condition, i.e., a 2D torus $\mathbb{T}$. If the lattice is $N \times N$, then $d = N^2$. Let $\vz \in \mathbb{T}$ be the 2D coordinate of a pixel in $\mathbb{T}$, then each grid cell $\evv_i$ is placed on a unique $\vz_i \in \mathbb{T}$. 

A CANN model hand-crafts the non-linear recurrent transformation $\vv(\vx+\Delta \vx) = F(\vv(\vx), \Delta \vx)$ for some parametric form of $F$, and the coding manifold $(\vv(\vx), \forall \vx)$ consists of the attracting fixed points of $F(\cdot, 0)$. In CANN, the recurrent connection weights between a pair of grid cells $(\evv_i, \evv_j)$ only depend on the relative positions of the two cells on the 2D torus, $\vz_i - \vz_j$, i.e.,  the connection weights are convolutional. Such a topographical arrangement may be physically realized on the 2D surface of the cortex in the brain, but it may also be the conceptual interpretation of the connection weights between the grid cells that are not necessarily placed on a physical 2D torus in the brain. 

If we place the grid cells $\vv = (\evv_i, i = 1, ..., d)$ on the $d = N \times N$ lattice of the 2D torus, either physically or conceptually, then their activities $\vv(\vx) = (\evv_i(\vx), i = 1, ..., d)$ form an $N \times N$ ``image'' defined on the 2D torus. The pattern of the ``image'' may be a localized ``bump'', i.e., only a local subset of the pixels of the $N \times N$ lattice have non-zero activities. Suppose each self-position $\vx$ of the agent can be mapped to a ``bump'' on the 2D torus centered at a corresponding $\vz \in \mathbb{T}$. When the agent moves in the 2D physical space, i.e., when $\vx$ changes to $\vx + \Delta \vx$, then the ``bump'' formed by $\vv(\vx) = (\evv_i(\vx), i = 1, ..., d = N^2)$ moves on the 2D torus from $\vz$ to $\vz + \Delta \vz$, while the shape of the ``bump'' remains the same. The connection weights of the CANN are hand-crafted so that the recurrent transformation of CANN realizes such a ``mirroring''  movement of the ``bump''. 

If each displacement $\Delta x$ of the agent in the 2D physical space can be mapped to a displacement $\Delta \vz$ of the ``bump'' on the 2D torus $\mathbb{T}$, then the recurrent transformation of the CANN forms a representation of the 2D Euclidean group $\mathbb{R}^2$. If the local movement of the ``bump'' $\delta \vz$ on the 2D torus is furthermore conformal to the local movement $\delta \vx$ of the agent in the 2D physical space, then the local movement of the $d = N^2$ dimensional vector $\vv(\vx) = (\evv_i(\vx), i = 1, ..., d = N^2)$ formed by the grid cells in the $d$-dimensional neural space, i.e., $\vv(\vx+\delta \vx) - \vv(\vx)$, is also conformal to the movement $\delta \vx$ of the agent in the 2D physical space, and the isotropic scaling condition also holds. 

In the above understanding, there are three types of movements. (1) The movement $\Delta \vx$ of the agent in the 2D physical space $\mathbb{R}^2$. (2) The movement $\Delta \vz$ of the ``bump'' on the $N \times N$ lattice of 2D torus $\mathbb{T}$. (3) The movement $\vv(\vx+\Delta \vx) - \vv(\vx)$ in the $d = N^2$-dimensional neural space.  

Our model on either the general transformation or the linear transformation does not assume a 2D torus topography. In fact, no 2D topographical structure whatsoever is assumed in our model. The topographical arrangement is not part of our model. Instead, it may be treated as an implementation issue after the model is learned, i.e., how to arrange the grid cells physically on a 2D surface of cortex so that a pair of grid cells with strong connection weights are placed close to each other. It may also be treated as an interpretation issue after the model is learned, i.e., how to interpret the learned connection weights.   

Even though our model does not make topographical assumptions, our linear transformation model appears to learn the torus topography automatically. Specifically, in our learned model, the response maps of the grid cells within each module are spatially shifted versions of the same hexagon periodic pattern. Therefore we can identify two directions in the 2D physical space that are $2\pi/3$ apart, so that $\vv(\vx)$ rotates back to itself as $\vx$ moves along these two directions for a certain distance. This implies that the codebook manifold $(\vv(\vx), \forall \vx)$ forms a 2D torus as assumed by CANN models. Moreover, the fact that the learned response maps of the grid cells within each module are spatially shifted versions of the same hexagon periodic pattern also agrees with the CANN model that moves the ``bump'' on the 2D torus by ``mirroring'' the motion in the 2D physical space. The learned hexagon periodic patterns and the spatial shifts of the response maps may be related to the optimality of the hexagon grid in terms of sampling, interpolation and packing. 

Even though the CANN model realizes the movement of the ``bump'' on the 2D torus by a non-linear recurrent model, such movement is a cyclic permutation of the activities of the grid cells, and the permutation can be realized by a permutation matrix, which is an orthogonal matrix. Thus the $\vv(\vx)$ that satisfies the non-linear CANN model also satisfies our linear transformation model, where the linear rotation matrix is a cyclic permutation matrix. 

The torus topology is hardly surprising, even for the general transformation model. The Lie group formed by $(F(\cdot, \Delta \vx), \forall \Delta \vx)$ is abelian as it is a representation of the 2D additive Euclidean group $\mathbb R^2$. If a connected abelian Lie group is compact, then the group is automatically a torus. See \cite{dwyer1998elementary}. 

Furthermore, if the scaling factor $s$ is globally a constant for all $\vx$, then the position embedding $(\vv(\vx), \forall \vx)$ is an isometric embedding up to a global scaling factor, and its intrinsic geometry remains Euclidean. It thus is a {\bf flat torus}. 

\section{Experiments}
\subsection{Implementation details} \label{supp:MC}
\paragraph{Monte Carlo samples.} The expectations in loss terms are approximated by Monte Carlo samples. Here we detail the generation of Monte Carlo samples. For $(\vx, \vx')$ used in $L_0 = \E_{\vx, \vx'}[A(\vx, \vx') - \langle \vv(\vx), \vu(\vx')\rangle]^2$, $\vx$ is first sampled uniformly within the entire domain, and then the displacement $d\vx$ between $\vx$ and $\vx'$ is sampled from a normal distribution $\mathcal{N}(0, \sigma^2 \mI_2)$, where $\sigma=0.48$. This is to ensure that nearby samples are given more emphasis. We let $\vx' = \vx + d\vx$, and those pairs $(\vx, \vx')$ within the range of domain (i.e., $1$m $\times$ $1$m, $40 \times 40$ lattice) are kept as valid data. For $(\vx, \Delta \vx)$ used in $L_1 = \E_{\vx, \Delta \vx}  |\vv(\vx+\Delta \vx) - \exp(\mB(\theta) \Delta r) \vv(\vx)|^2$, $\Delta \vx$ is sampled uniformly within a circular domain with radius equal to $3$ grids and $(0, 0)$ as the center. Specifically, $\Delta r^2$, the squared length of $\Delta \vx$, is sampled uniformly from $[0, 3]$ grids, and $\theta$ is sampled uniformly from $[0, 2\pi]$. We take the square root of the sampled $\Delta r^2$ as $\Delta r$ and let $\Delta \vx = (\Delta r \cos \theta, \Delta r \sin \theta)$. Then $\vx$ is uniformly sampled from the region such that both $\vx$ and $\vx + \Delta \vx$ are within the range of domain. For $(\theta, \Delta \theta)$ used in $L_2 = \sum_{k=1}^{K} \E_{\vx, \theta, \Delta \theta} [\|\mB_k(\theta + \Delta \theta) \vv_k(\vx)\| - \|\mB_k(\theta) \vv_k(\vx)\| ]^2$, we uniformly sample $\theta$ and $\theta + \Delta \theta$ from discretized angles, i.e., $144$ directions discretized for circle $[0, 2\pi]$. We will study sampling only small $\Delta \theta$ in the future. 

\paragraph{Training details.} The model is trained for 14,000 iterations. At each iteration, the samples are generated online. For the first 8,000 iterations, we update all learnable parameters, while for the following iterations, we fix the learned $\vv(\vx)$ and update the other learnable parameters. The initial learning rate is set as $0.003$ and is decreased by a factor of 0.5 every 500 iterations after 8,000 iterations. We use Adam~\cite{kingma2014adam} optimizer. The model is trained on a single Titan XP GPU. We apply the maximum batch size that can fit into the single GPU, which is 90,000. It takes about 3.5 hours to train the model on a single Titan XP GPU. 
 
\paragraph{Baseline methods.} In Table 1 of the main text, we compare the learned neurons with the ones from other two optimization-based learning methods \cite{banino2018vector,sorscher2019unified}. For \cite{banino2018vector}, we run the code released by the authors (\url{https://github.com/deepmind/grid-cells}) to learn the model and compute gridness scores for the learned neurons. For \cite{sorscher2019unified}, we use the pre-trained weights released by the authors (\url{https://github.com/ganguli-lab/grid-pattern-formation}) to get the learned neurons and compute the gridness scores. Both the code of \cite{banino2018vector} and pre-trained weights of \cite{sorscher2019unified} use Apache License V2. 

\paragraph{Usage of data.} In this paper, we mainly use simulated trajectories as training data, and thus we do not think that the data contain any personally identifiable information or offensive content. The only existing data we use is the pre-trained weights of the baseline method \cite{sorscher2019unified}. Under Apache License V2, we believe it is fully approved by the authors to use the pre-trained weights. 

 \subsection{Learned patterns}
 
Fig. \ref{fig:autocorr} displays the autocorrelograms of learned patterns of $\vv(\vx)$.

\begin{figure}[h]
\centering
  \includegraphics[width=\linewidth]{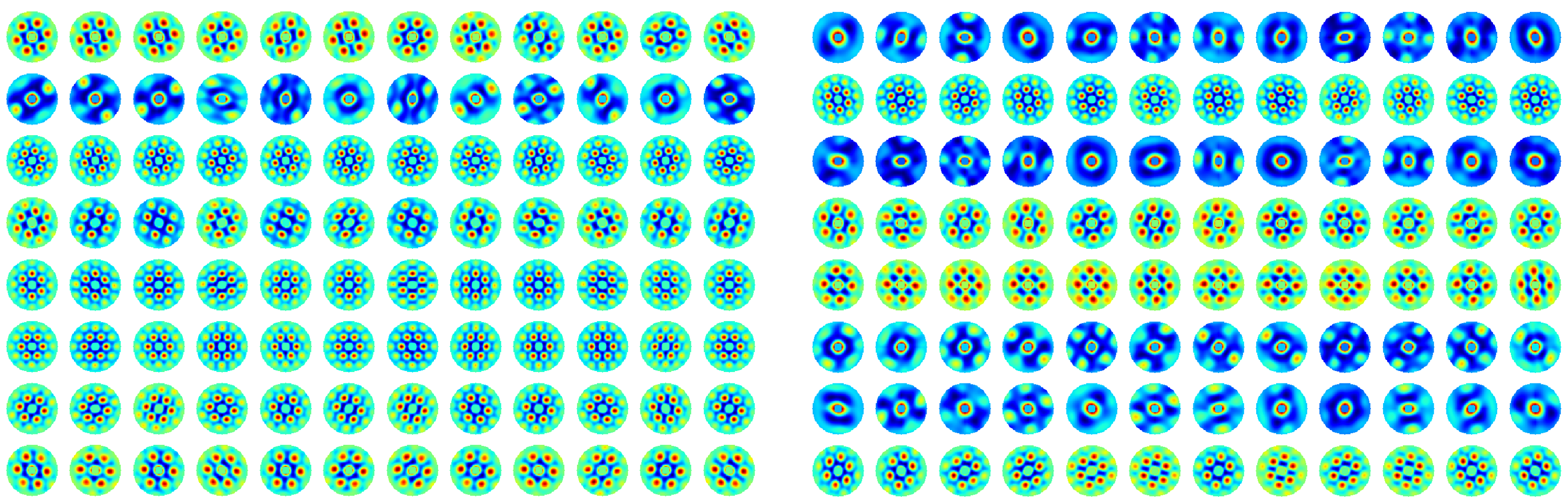}
  \caption{\small Autocorrelograms of the learned patterns of $\vv(\vx)$.}
  \label{fig:autocorr}
\end{figure} 

Fig. \ref{fig: u} shows the learned patterns of $\vu(\vx)$ with $16$ blocks of  $12$ cells in each block. Regular hexagon patterns also emerge. 
 
\begin{figure}[h]
\centering
  \includegraphics[width=\linewidth]{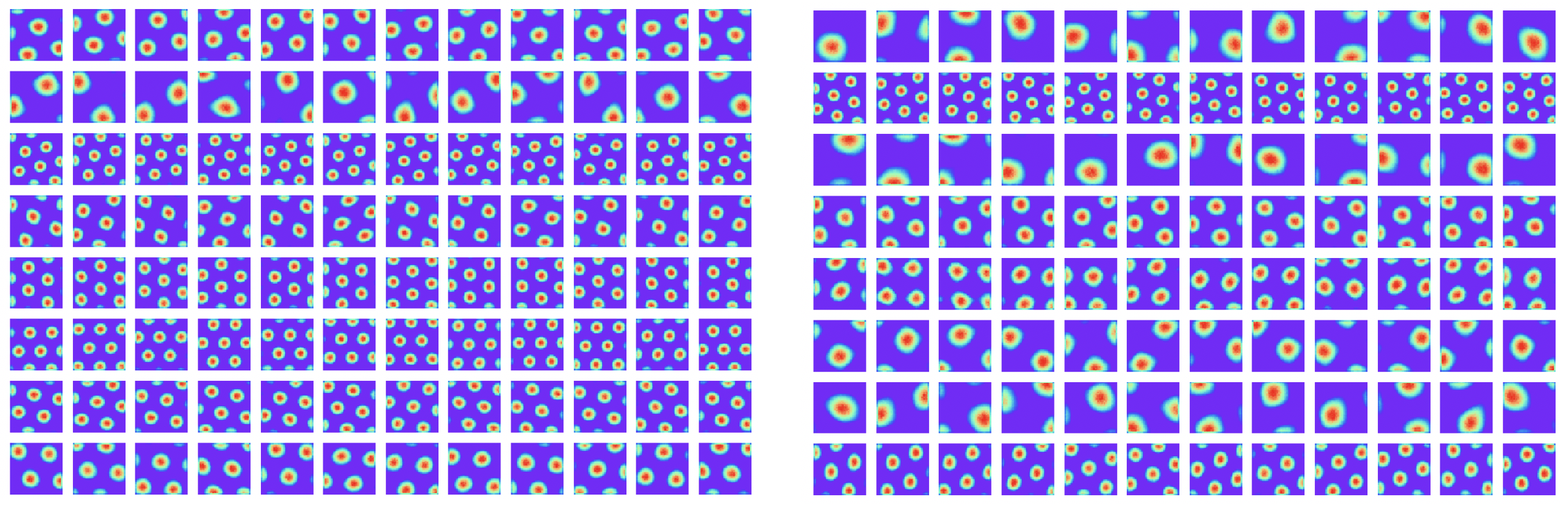}
  \caption{\small Learned patterns of $\vu(\vx)$ with $16$ blocks of size $12$ cells in each block. Every row shows the learned patterns within the same block.}
  \label{fig: u}
\end{figure}

For learned firing patterns of $\vv(\vx)$, we also display the histogram of grid orientations in Fig. \ref{fig:orientation}, where we do not observe clear clusters. 

\begin{figure}[h]
\centering
  \includegraphics[width=.3\linewidth]{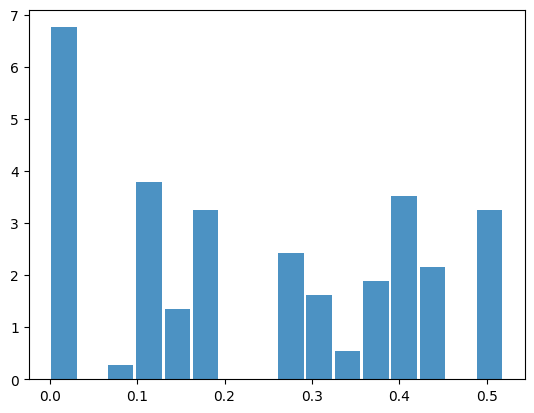}
  \caption{\small Histogram of grid orientations of the learned firing patterns of $\vv(\vx)$. }
  \label{fig:orientation}
\end{figure}

In Fig. \ref{fig:B}, we show the learned patterns of a block of $\mB(\theta)$. Each element shows significant sine/cosine tuning over $\theta$. For the other blocks, the patterns are all similar.  

\begin{figure}[h]
\centering
  \includegraphics[width=.7\linewidth]{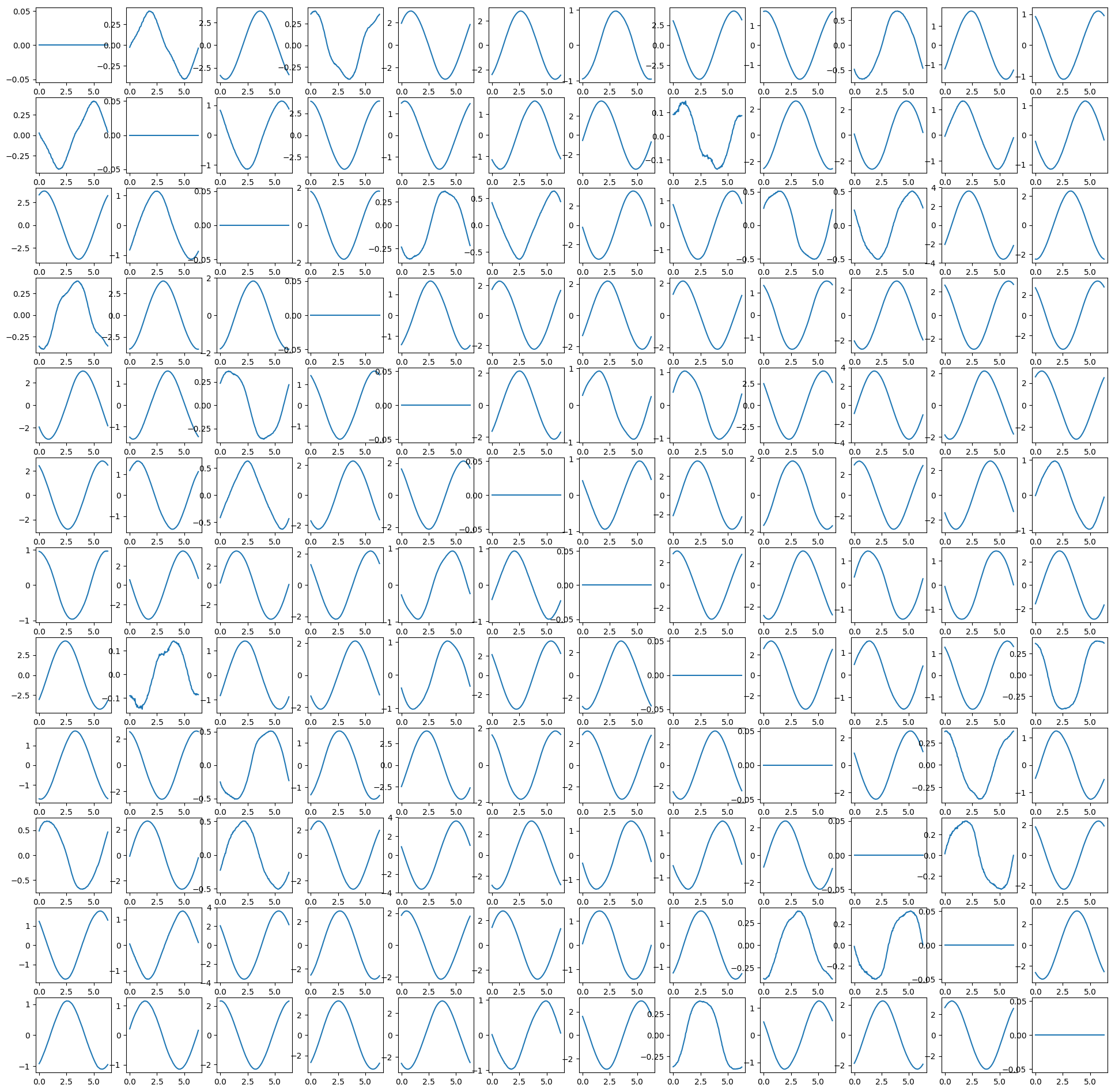}
  \caption{\small Learned patterns of a block of $\mB(\theta)$. Each subfigure shows the value of an element in $\mB(\theta)$ (vertical axis) over $\theta$ (horizontal axis).}
  \label{fig:B}
\end{figure}

\paragraph{Gaussian kernel.} Because $A(\vx, \vx')$ is a sharp Gaussian kernel, it contains a whole range of frequencies in the 2D Fourier domain. The learned response maps of the grid cells span a range of frequencies or scales too.  Each module or block focuses on a certain frequency band, which corresponds to the metric of the module.  We assume individual place field $A(\vx, \vx')$  to exhibit a Gaussian shape, rather than a Mexican-hat pattern (with balanced excitatory center and inhibitory surround) as assumed in previous basis expansion models ~\citep{dordek2016extracting,sorscher2019unified} of grid cells. The Mexican-hat or difference of Gaussians pattern occupies a ring in the 2D Fourier domain. It corresponds to a module in our model. But we use isotropic condition to enforce each module to be within a ring in the Fourier domain, and we use different modules to pave the whole Fourier domain.

\subsection{Error correction} 

We begin by assessing the ability of error correction of the learned system following the setting in Proposition 1. Specifically, for a given location $\vx$, suppose the neurons are perturbed by Gaussian noise: $\vv = \vv(x) + {\bf \epsilon}$, where ${\bf \epsilon} \sim \mathcal{N}(0, \tau^2 \mI_d)$ and $\tau^2 = \alpha^2 (\|\vv(\vx)\|^2/d)$, so that $\alpha^2$ measures the variance of noise relative to the average magnitude of $(v_i(\vx)^2, i = 1, ..., d)$ and $\alpha$ measures the relative standard deviation. We infer the 2D position $\hat{\vx}$ from $\vv$ by $\hat{\vx} = \arg\min_{\vx'}\|\vv - \vv(\vx')\|^2$. Fig. \ref{fig: error_correction0} displays the inference error over the relative standard deviation $\alpha$ of the added Gaussian noise. We also show the results using the learned $\vu(\vx')$ for inference (Eq. (\ref{eq:de2})). The system works remarkably well even if $\alpha = 2$. 

\begin{figure}[ht]
\centering
  \includegraphics[width=.4\linewidth]{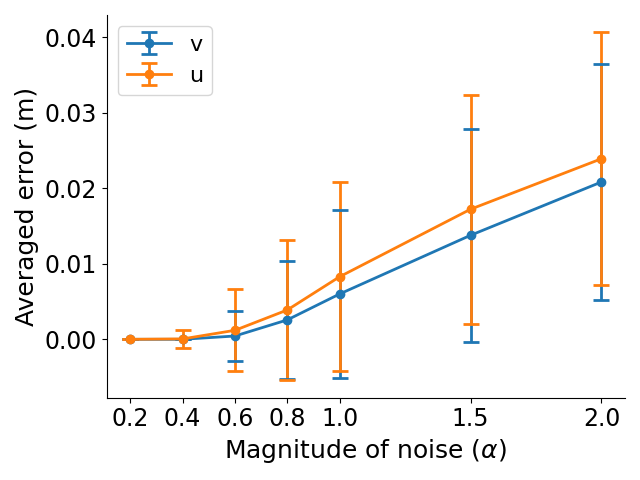}
  \caption{\small Error correction results following the setting in Proposition 1. The error bar stands for the standard deviation over 1,000 trials. ``$\vv$'' means decoding by Eq. (\ref{eq:de1}), and ``$\vu$'' means decoding by Eq. (\ref{eq:de2}). The squared domain is $1$m $\times$ $1$m. }
  \label{fig: error_correction0}
\end{figure}
	
	We further assess the ability of error correction in long distance path integration. Specifically, along the way of path integration, at every time step $t$, two types of errors are introduced to $\vv_t$: (1) Gaussian noise or (2) dropout masks, i.e., certain percentage of units are randomly set to zero. Fig. \ref{fig: error_correction} summarizes the path integration performance with different levels of injected errors for $T=100$, using $\vv(\vx')$ (Eq. (\ref{eq:de1})) or $\vu(\vx')$ (Eq. (\ref{eq:de2})) for decoding.  The results show that re-encoding at each step helps error correction, especially for dropout masks. For Gaussian noise, even without decoding and re-encoding at each step, decoding at the final step alone is capable of removing much of the noise. 	
	Notably, with re-encoding, the path integration works well even if Gaussian noise with $\alpha = 1$ is added or $50\%$ units are randomly dropped out at each step, indicating that the learned system is robust to different sources of errors.  
	
\begin{figure}[ht]
\centering
  \includegraphics[width=.49\linewidth]{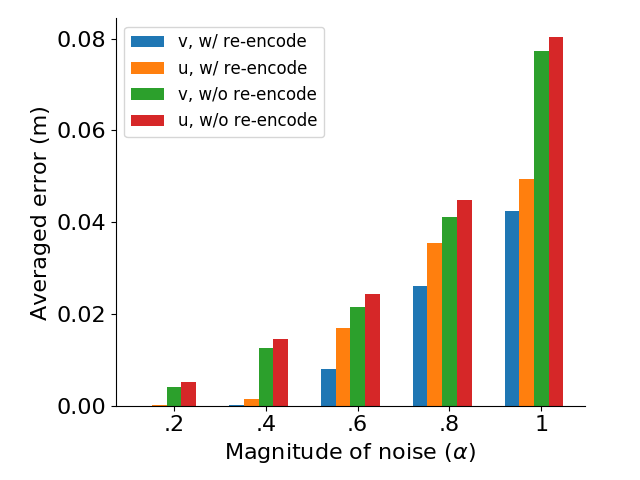}
    \includegraphics[width=.49\linewidth]{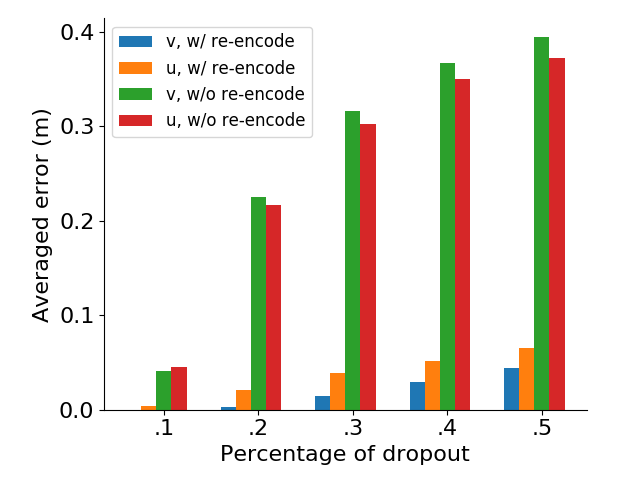}
  \caption{\small Path integration results with different levels of injected errors. {\em Left}: Gaussian noise. The magnitude of noise is measured using the average of the squared magnitudes of the units in $\vv(\vx)$ as the reference. {\em Right}: dropout masks. Certain percentage of units are randomly set to zero at each step. ``$\vv$'' means decoding by Eq. (\ref{eq:de1}), and ``$\vu$'' means decoding by Eq. (\ref{eq:de2}).  The squared domain is $1$m $\times$ $1$m.}
  \label{fig: error_correction}
\end{figure}

\subsection{Non-linear transformation model} 

We test our method with a non-linear transformation model:
\begin{eqnarray}
F(\vv(\vx), \Delta r, \theta) = {\rm ReLU}(\exp(\mB(\theta) \Delta r) \vv(\vx)),    \label{eq:nonlinear}
\end{eqnarray}
where we insert ${\rm ReLU}(a) = \max(0, a)$  into the linear transformation model. 

We use numerical differentiation to define directional derivative
\begin{eqnarray}
   f_\theta(\vv(\vx)) = [\vv(\vx + \delta \vx) - \vv(\vx)]/\delta r, 
\end{eqnarray} 
where $\delta \vx = (\delta r \cos \theta, \delta r \sin \theta)$, with pre-defined $\delta r$. The reason for numerical differentiation is because the derivative of ReLU is an indicator function, which is not differentiable. $f_\theta(\vv(\vx))$ needs to be differentiable for minimizing the loss function (an alternative to numerical differentiation is to use sigmoid function to approximate the indicator function). 

We continue to use the same loss function except with the above two changes. Interestingly, regular hexagon patterns continue to emerge (average gridness score 0.83, percentage of grid cells 70.21$\%$). See Fig. \ref{fig: nonlinear} for the learned patterns of $\vv(\vx)$.

\begin{figure}[h]
\centering
  \includegraphics[width=\linewidth]{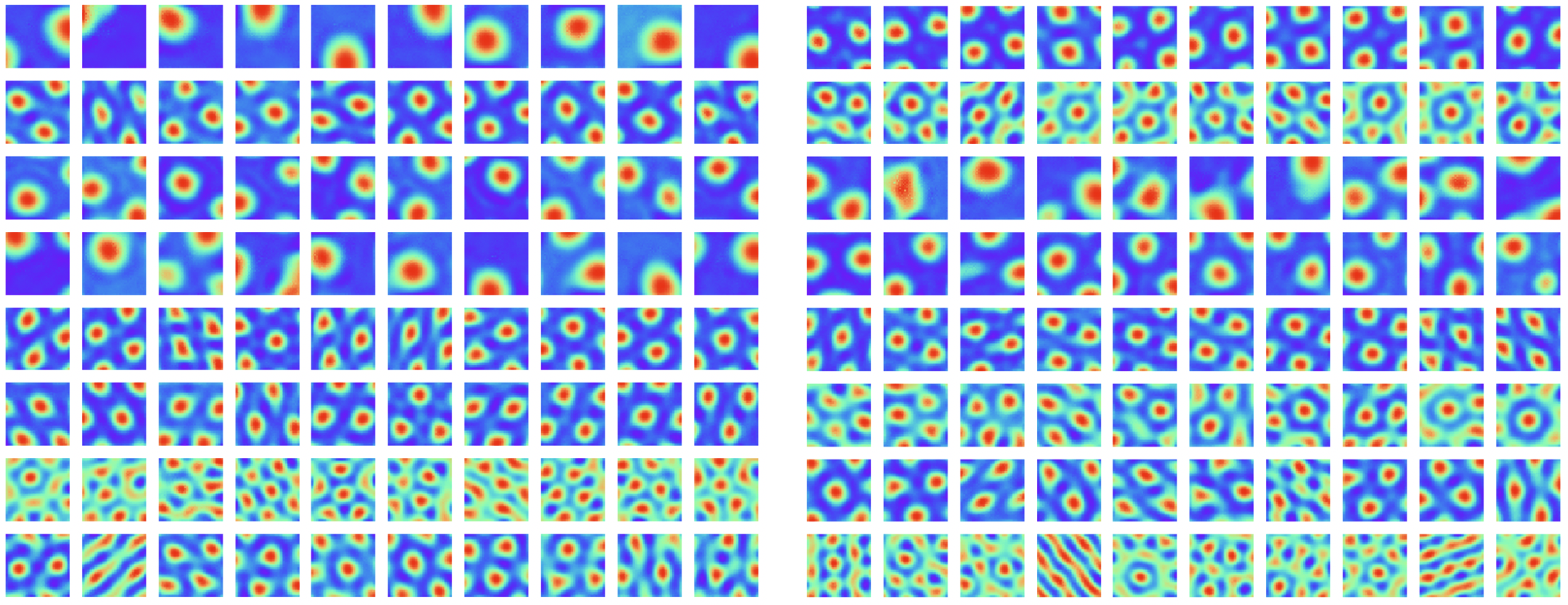}
  \caption{\small Learned patterns of $\vv(\vx)$ with the non-linear transformation model (Eq. (\ref{eq:nonlinear})). Every row shows the learned patterns within the same block.}
  \label{fig: nonlinear}
\end{figure}

\subsection{Path planning} 

Our grid cells model can be applied to path planning. Specifically, according to \cite{stachenfeld2017hippocampus}, the adjacency kernel can be modeled by
\begin{eqnarray}
 A_\gamma(\vx, \vx') = \mathbb{E} \left[ \sum_{t=0}^{\infty} \gamma^t 1(\vx_t = \vx') | \vx_0 = \vx)\right] = \langle \vv(\vx), \vu_\gamma(\vx')\rangle,
\end{eqnarray} 
where $\gamma$ is the discount factor that controls the temporal and spatial scales, $\mathbb{E}$ is with respect to a random walk exploration policy, and $1(\cdot)$ is the indicator function. For random walk in open field, $A_\gamma(\vx, \vx') \propto \exp(-\|\vx-\vx'\|^2/2\sigma_{\gamma}^2)$, where $\sigma_{\gamma}^2$ depends on $\gamma$. 

To enable path planning, we need kernels of both big and small spatial scales to account for long and short distance planning respectively. To this end, we discretize $\gamma$ into a finite list of scales, and learn a list of corresponding $\vu_\gamma(\vx’)$ together with $\vv(\vx)$ and $\mB(\theta)$ using the loss function in Section 5 of the main text. 

With the learned model, path planning can be accomplished by steepest ascent on the adjacency to the target position. Specifically, let $\hat{\vx}$ be the target or destination. Let $\vx^{(t)}$ be the current position in the path planning process, encoded by $\vv(\vx^{(t)})$. The agent plans the next displacement by steepest ascent on 
\begin{eqnarray}
A_\gamma(\vx^{(t)} + \Delta \vx, \hat{\vx}) = \langle \vv(\vx^{(t)} + \Delta \vx), \vu_\gamma(\hat{\vx})\rangle
=\langle \mM(\Delta \vx) \vv(\vx^{(t)}), \vu_\gamma(\hat{\vx})\rangle, 
\end{eqnarray}
over allowed $\Delta \vx$ within a single step, where $\mM(\Delta \vx) = \exp(\mB(\theta) \Delta r)$, with $\Delta \vx = (\Delta r \cos \theta, \Delta r \sin \theta)$.
We plan 
\begin{eqnarray}
\Delta \vx^{(t+1)} = \arg \max_{\Delta \vx } A_\gamma(\vx^{(t)} + \Delta \vx, \hat{\vx}),  \label{eq:plan}
\end{eqnarray}
and let $\vx^{(t+1)} = \vx^{(t)} + \Delta \vx^{(t+1)}$.  

The scale $\gamma$ is selected as the smallest one that satisfies $ \max_{\Delta \vx } \langle \mM(\Delta \vx) \vv(\vx^{(t)}), \vu_\gamma(\hat{\vx})\rangle > .2$. We can also use $\max_\gamma \max_{\Delta \vx } \langle \mM(\Delta \vx) \vv(\vx^{(t)}), \vu_\gamma(\hat{\vx})\rangle$ for scale selection. 

We test path planning in the open field environment. The model is first learned using a single-scale kernel function $A_\gamma(\vx, \vx’) = \exp(-\|\vx-\vx'\|^2/2\sigma_{\gamma}^2)$ where $\sigma_\gamma = 0.07$. Then we assume a list of three scales: $\sigma_\gamma = [0.07, 0.14, 0.28]$ and learn the corresponding list of $\vu_\gamma(\vx’)$. The pool of allowed displacements for a single step is defined as:  $d r$ can be 1 or 2 grids, while $\theta$ can be chosen from 200 discretized angles over $[0, 2\pi]$. Fig. \ref{fig: planning} demonstrates several examples of path planning in the open field environment, where the agent is able to plan straight path to the target. When $\vx^{(t)}$ is far from the target, kernel with large $\sigma_\gamma$ is chosen, and as $\vx^{(t)}$ approaches the target, the chosen kernel gradually switches to the one with small $\sigma_\gamma$. A planning episode is treated as a success if the distance between $\vx^{(t)}$ and target is smaller than $0.5$ grid within $40$ time steps. The agent achieves a success rate of 100$\%$ (tested for $10,000$ episodes). 

\begin{figure}[h]
\centering
  \includegraphics[width=.24\linewidth]{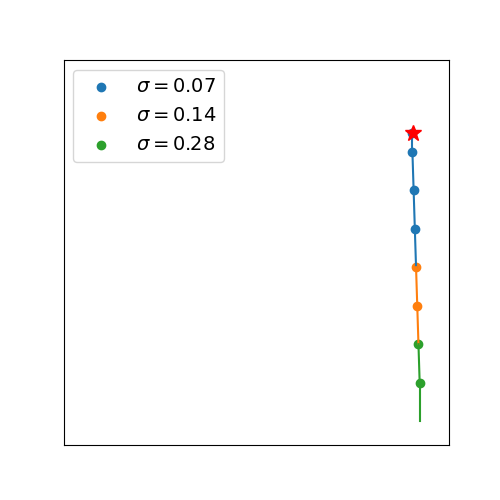}
  \includegraphics[width=.24\linewidth]{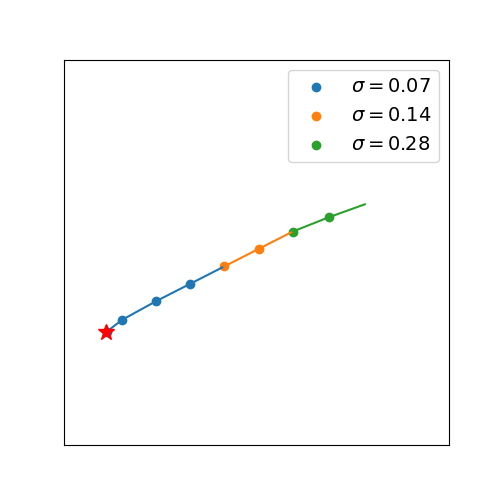}
  \includegraphics[width=.24\linewidth]{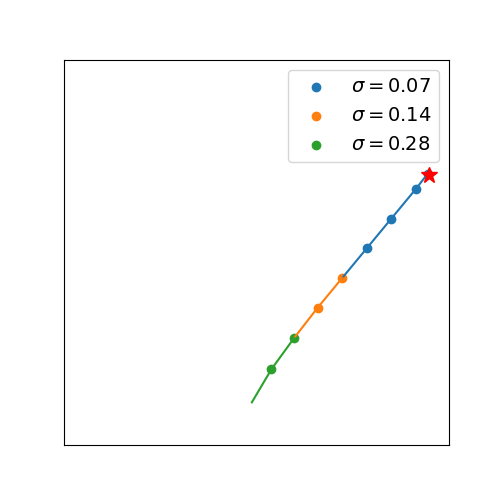}
  \includegraphics[width=.24\linewidth]{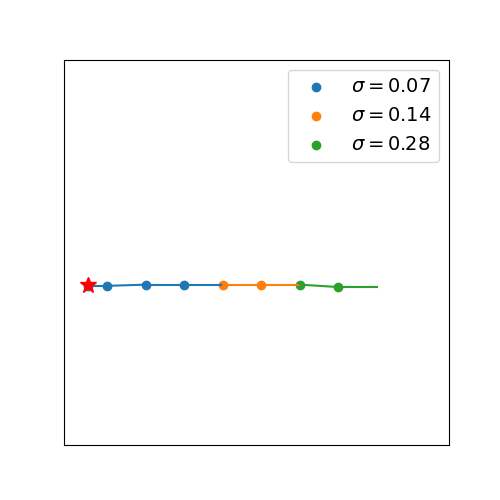}
  \caption{\small Examples of path planning results in an open field environment. The target is shown as a red star. }
  \label{fig: planning}
\end{figure}

For a field with obstacles or rewards, we can learn the deformed $A_\gamma(\vx, \vx’)$ and $(\vv(\vx), \vu_\gamma(\vx’))$ by temporal difference learning with a random walk exploration policy as suggested in \cite{stachenfeld2017hippocampus}. After learning $A_\gamma(\vx, \vx’)$ and $(\vv(\vx), \vu_\gamma(\vx’))$, we can continue to use Eq. (\ref{eq:plan}) for path planning. We shall further study it in future work.

\subsection{Integrating egocentric vision} 

When the agent moves in darkness, it can infer its self-position by integrating self-motion, as illustrated by our experiments on path integration. If there is visual input, the agent can infer its self-position (as well as head direction) from the visual image alone. We extend our grid cells model to study this problem of egocentric vision, which is important in computer vision. 

Specifically, suppose the agent navigates in a 3D scene such as a room, and the height of the eye (or camera) remains fixed. Suppose at 2D self-position $\vx$ and with head direction $\theta$, the agent sees an image $\mI$, which is called a posed image. We use the vector representation $\vv(\vx)$ in our original grid cells model to represent the 2D self-position $\vx$, and use another vector representation $\vh(\theta)$ to represent the head direction $\theta$. If the agent changes its head direction from $\theta$ to $\theta + \Delta \theta$, $\vh(\theta)$ is transformed to
\begin{eqnarray} 
\vh(\theta + \Delta \theta) = \exp(\mC\Delta\theta)  \vh(\theta).
\end{eqnarray}
We assume that there are $K$ modules or blocks in $\vh(\theta)$ and $\mC$ is skew-symmetric. This is similar to the transformation of $\vv(\vx)$ in our grid cells model. 

$(\vx, \theta)$ is called the pose of the camera (or eye), and we call $(\vv(\vx), \vh(\theta))$ the pose embedding. 

To associate the pose embedding $(\vv(\vx), \vh(\theta))$ with the posed image $\mI$, we use a vector representation or scene embedding $\vs$ to represent the 3D scene which is shared across different posed images of the same scene, and we learn a generator network $G_\beta$ that maps the embeddings $\vs$ and $(\vv(\vx), \vh(\theta))$ to the posed image $\mI$:
\begin{eqnarray} 
\mI = G_\beta(\vs, \vv(\vx), \vh(\theta)) + \epsilon,
\end{eqnarray} 
where the generator $G_\beta$ is parametrized by a multi-layer deconvolutional neural network with parameters $\beta$, and $\epsilon$ is the residual error. 

Given the above assumptions, we introduce two extra loss terms in addition to the loss function described in Section 5 of the main text. 
\begin{eqnarray}
&&L_3 = \sum_{k=1}^K \mathbb{E}_{\theta, \Delta \theta} \|\vh_k(\theta + \Delta \theta) - \exp(\mC_k \Delta \theta)\vh_k(\theta)\|^2,\\
&&L_4 = \mathbb{E} \|\mI - G_\beta(\vs, \vv(\vx), \vh(\theta))\|^2.
\end{eqnarray} 
$L_3$ is to model the head rotation, and $L_4$ is to model the generation of the posed image. 

During training, we alternatively update $(G_\beta, \vs)$ and $(\vv(\vx), \mB(\theta), \vu(\vx’), \vh(\theta), \mC)$ by gradient descent on the overall loss function that is a linear combination of $L_0$, $L_1$ and $L_2$ in the main text, as well as $L_3$ and $L_4$ introduced above. 

The learned model enables two useful applications: 

(a) {\bf Novel view synthesis}. Given an unseen pose $(\vx, \theta)$, the model can predict the corresponding posed image by $G_\beta(\vs, \vv(\vx), \vh(\theta))$.

(b) {\bf Inference of pose}, i.e., self-position $\vx$ and head direction $\theta$, from posed image $\mI$ alone. Specifically, after training the model, we can learn an additional inference network $F_\xi$ that maps an observed posed image $\mI$ to its pose embedding $\vv(\vx)$ and $\vh(\theta)$. The inference network is learned by minimizing the $\ell_2$ distance between the predicted and true pose embeddings: $\mathbb{E}\|(\vv(\vx), \vh(\theta)) - F_\xi(\mI)\|^2$. Then given an unseen posed image $\mI$, we can infer the pose by $\arg \min_{\vx, \theta}\|(\vv(\vx), \vh(\theta)) - F_\xi(\mI)\|^2$. In this task, $F_\xi(\mI)$
 is the estimate of $(\vv(\vx), \vh(\theta))$, and it is likely that this estimate contains error. This error will translate to the error in the estimated $(\vx, \theta)$. Thus our theoretical analysis of error translation in the main text is highly relevant, and the isotropic scaling condition is motivated by the analysis of error translation.

We conduct experiments on a dataset generated by the Gibson Environment \cite{xia2018gibson}, which provides tools for rendering images of different poses in 3D rooms. Specifically, we select 20 areas of size 2m $\times$ 2m from different rooms and render about 28k 64 $\times$ 64 RGB posed images for each area. The camera height is fixed and the camera can only rotate horizontally. The scene embedding vector $\vs$ is of 512 dimensions. Both $\vv(\vx)$ and $\vh(\theta)$ are of $192$ dimensions, partitioned into $K=16$ modules. 

Hexagon patterns still emerge in the learned $\vv(\vx)$ (average gridness score 0.71). For novel view synthesis, we evaluate the performance on 374k testing posed images. The resulting peak signal-to-noise ratio (PSNR) between synthesized images and ground truth images is 25.17, indicating that the model can generate reasonable unseen posed images. Fig. \ref{fig: 3d} demonstrates several examples of the novel view synthesis results. 

\begin{figure}[h]
\centering
  \includegraphics[width=.3\linewidth]{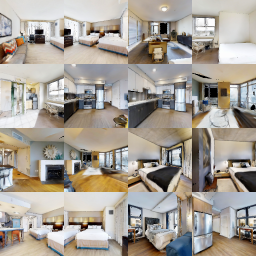} \hspace{3mm}
  \includegraphics[width=.3\linewidth]{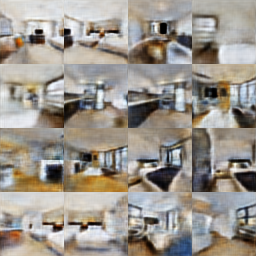}
  \caption{\small Examples of synthesizing novel views. {\em Left}: Ground truth unseen posed images. {\em Right}: synthesized unseen posed images. }
  \label{fig: 3d}
\end{figure}

For inference of pose (self-position $\vx = (\evx_1, \evx_2)$ and head direction $\theta$), we evaluate the performance on the same 374k testing posed images and report the average inference error in Table \ref{tab: inference}. The estimates are reasonably accurate. 
\begin{table}
	\begin{center}
\caption{\small Average error of pose inference.}
\label{tab: inference}
\begin{tabular}{ cccc }
\toprule
  & $\evx_1$ & $\evx_2$ & $\theta$ \\ \midrule
 Error & .0225m & .0230m & 1.37$^{\rm o}$ \\  
 \bottomrule
\end{tabular}
\end{center}
\end{table}

\subsection{Ablation studies}

\paragraph{Isotropic scaling condition is necessary for hexagon grid patterns.}
A natural question is whether the isotropic scaling condition (condition 2) is important for learning hexagon grid patterns. To verify this, we learn the model by removing the loss term $L_2$ (Eq. (19) in the main text) from the loss function, which constrains the model to meet condition 2. As shown in Fig. \ref{fig: ablation-isotropic}, more strip-like patterns emerge without $L_2$, indicating that condition 2 is important for hexagon grid patterns to emerge. 

\begin{figure}[h]
\centering
  \includegraphics[width=\linewidth]{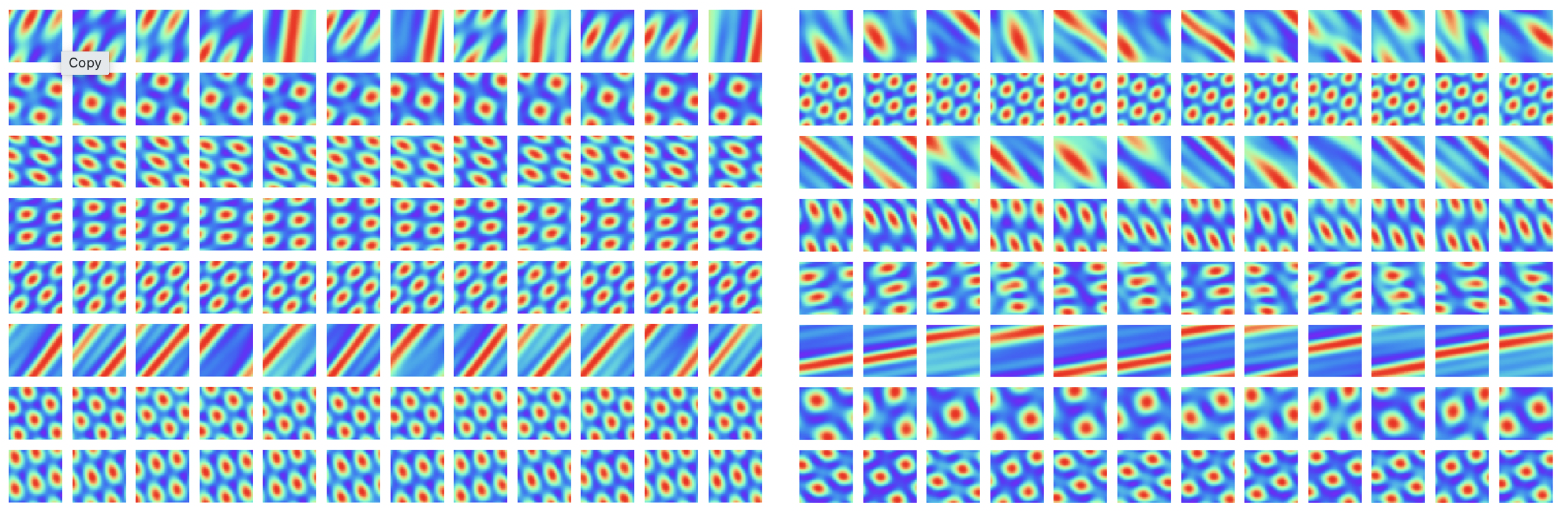}
  \caption{\small Learned neurons without loss term $L_2$, which is the constraint on isotropic scaling condition. More strip-like firing patterns emerge.}
  \label{fig: ablation-isotropic}
\end{figure}

\paragraph{Assumption of $\vu(\vx') \geq 0$ is not necessary for hexagon grid patterns.}
During training, we make an assumption of $\vu(\vx') \geq 0$ to make sure the connections from grid cells to place cells are excitatory \citep{zhang2013optogenetic,Rowland2018}. However, we want to emphasize this is not a key assumption in our model. Fig. \ref{fig: ablation-u} demonstrates the learned neurons in the network without assuming $\vu(\vx') \geq 0$, where hexagonal grid firing patterns also emerge. The average gridness score is 0.82 and the percentage of grid cells is $87.50\%$. However, the grid activations can be either positive/excitatory (in red color) or negative/inhibitory (in blue color). 

\begin{figure}[h]
\centering
  \includegraphics[width=\linewidth]{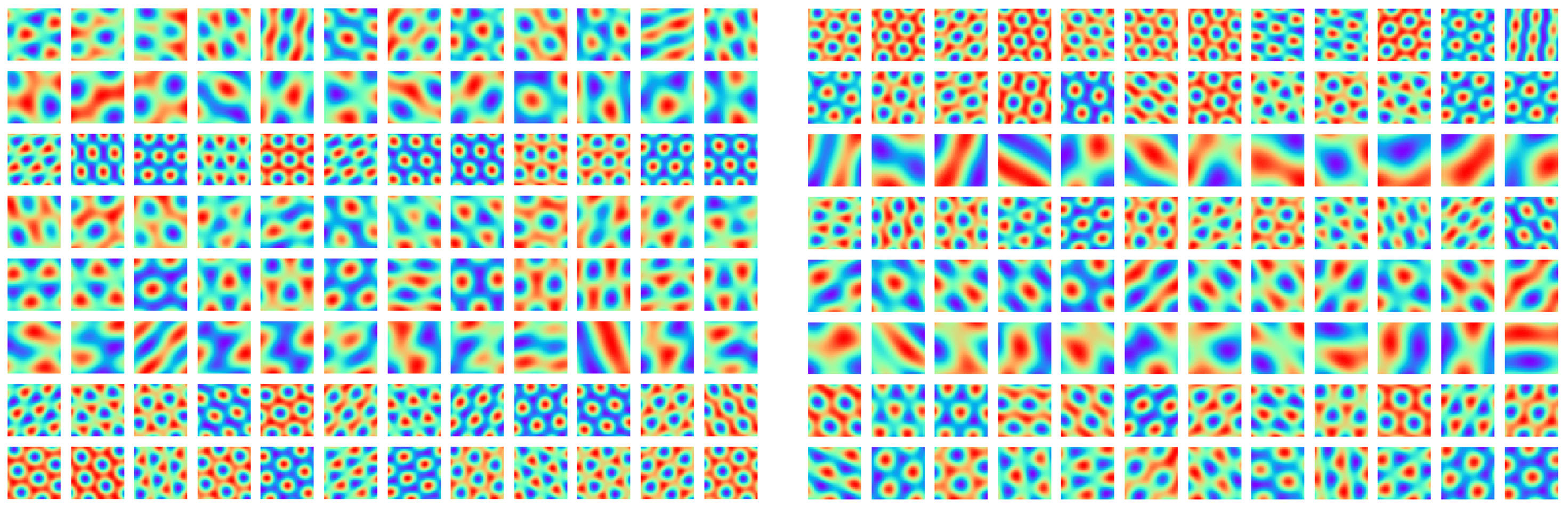}
  \caption{\small Learned neurons without the assumption of $\vu(\vx') \geq 0$. Hexagonal grid firing patterns also emerge, with the grid activations being either positive/excitatory (in red color) or negative/inhibitory (in blue color).}
  \label{fig: ablation-u}
\end{figure}

\paragraph{Skew-symmetric assumption of $\mB(\theta)$ is not important for hexagon grid patterns.} To make the linear transformation a rotation, we have assumed that $\mB(\theta)$ is skew-symmetric, i.e., $\mB(\theta) = -\mB(\theta)^\top$. Nonetheless, this assumption is not important for the emergence of hexagon grid patterns. Fig. \ref{fig:skew} demonstrates the learned neurons without assuming that $\mB(\theta)$ is skew-symmetric. Hexagon grid firing patterns emerge in most of the neurons, with only one block of square grid firing patterns. 

\begin{figure}[h]
\centering
  \includegraphics[width=\linewidth]{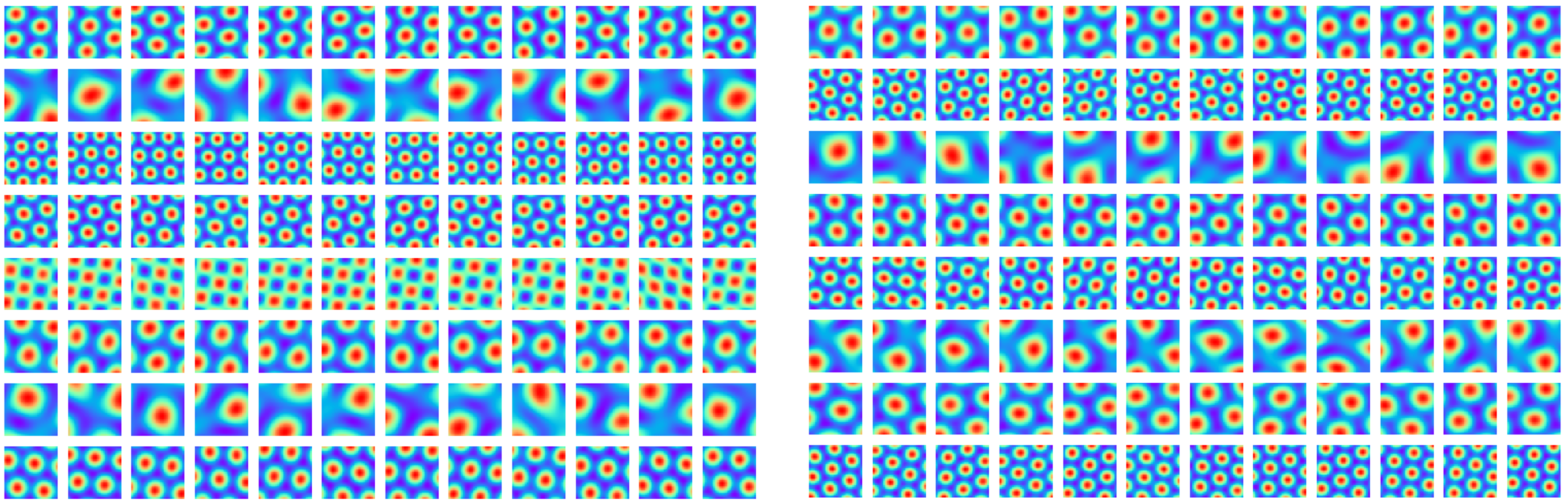}
  \caption{\small Learned neurons without skew-symmetric assumption of $\mB(\theta)$. Hexagonal grid firing patterns emerge in most of the neurons, with a block of square grid firing patterns.}
  \label{fig:skew}
\end{figure}

\paragraph{Number and sizes of blocks do not matter.} It is worthwhile to mention that the emergence of hexagonal grid firing patterns in the learned neurons are not due to specific design of the block size or the number of blocks. Fig. \ref{fig: units} visualizes the learned neurons by fixing the total number of neurons at $192$ and altering the block size and number of blocks. Hexagon patterns emerge in all the settings.

\begin{figure}[h]
	\centering
	\includegraphics[width=.8\linewidth]{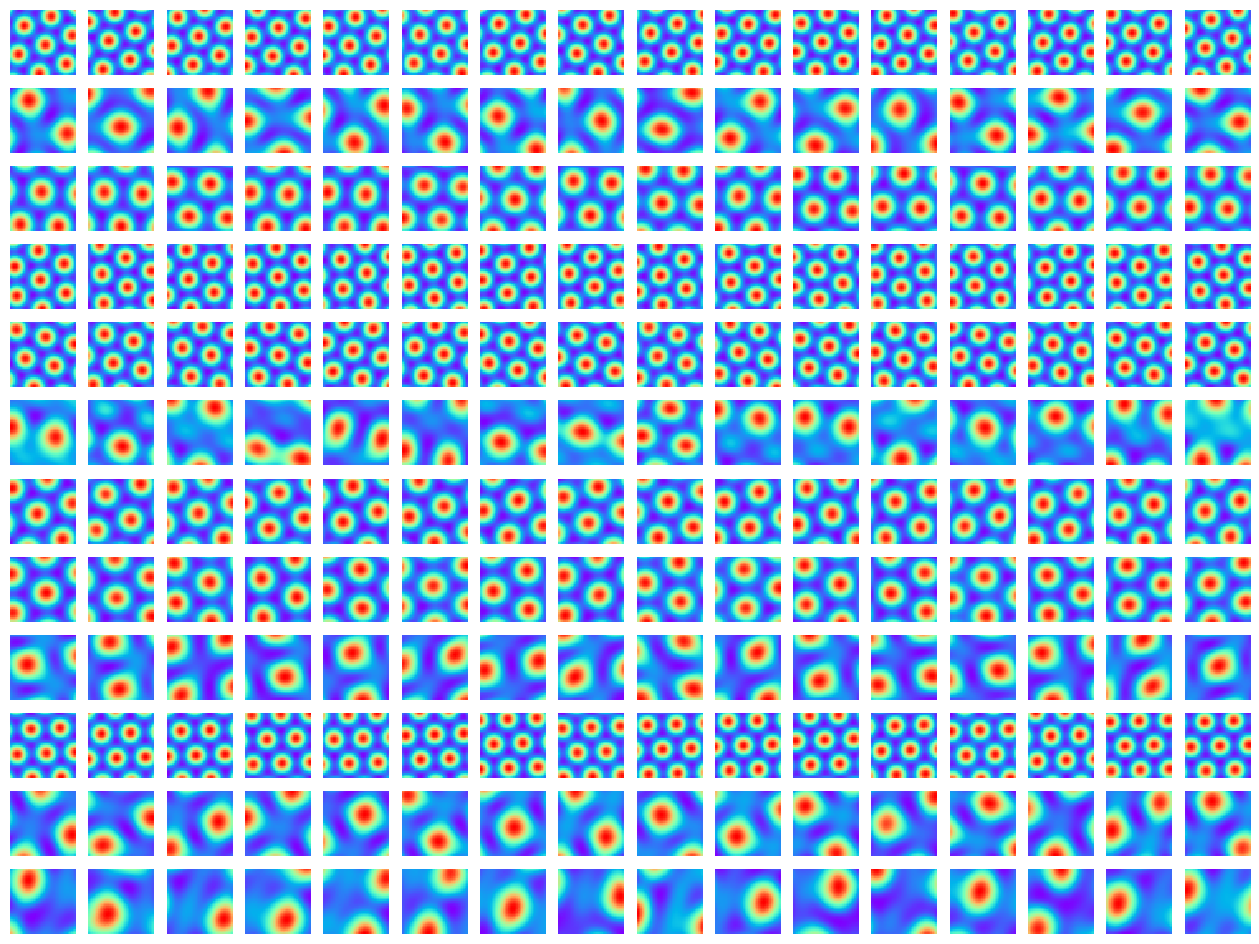}\\
	(a) Block size $= 16$ \vspace{1mm} \\
	\includegraphics[width=\linewidth]{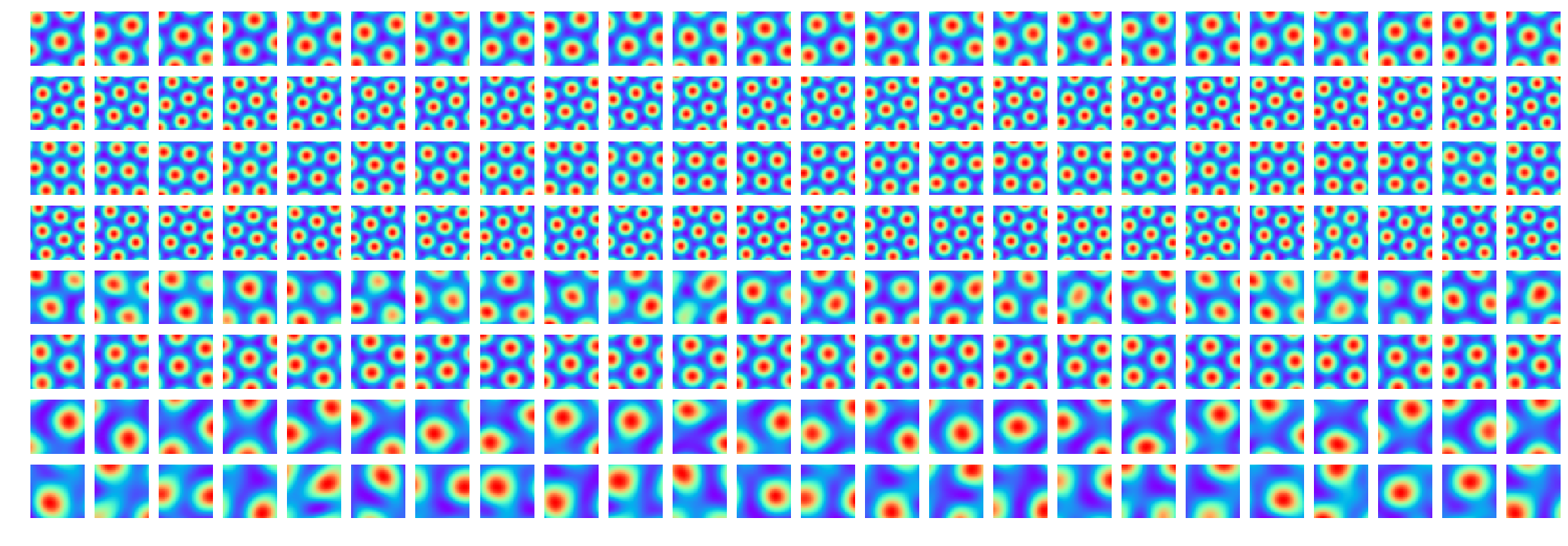}\\
	(b) Block size $= 24$ \vspace{1mm}\\
	\includegraphics[width=\linewidth]{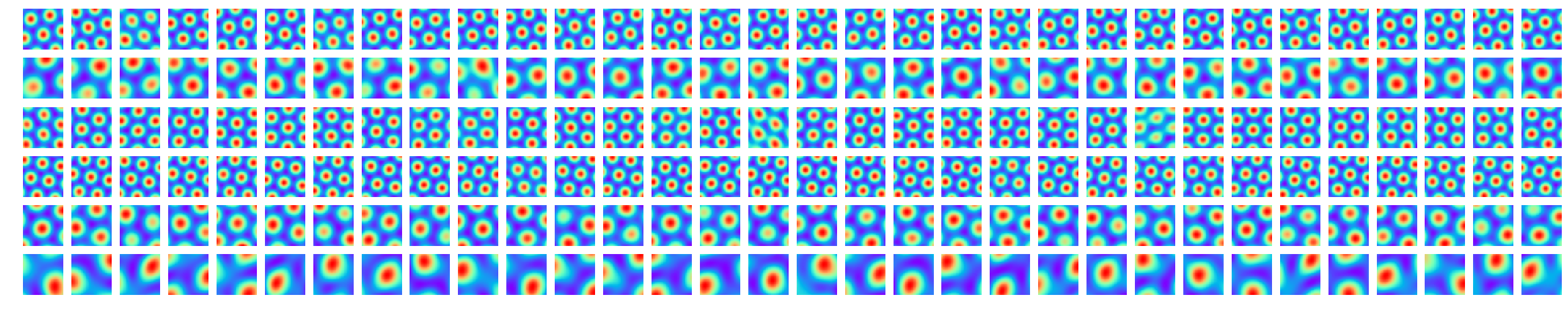}\\
	(c) Block size $= 32$ \vspace{1mm}\\
	\includegraphics[width=\linewidth]{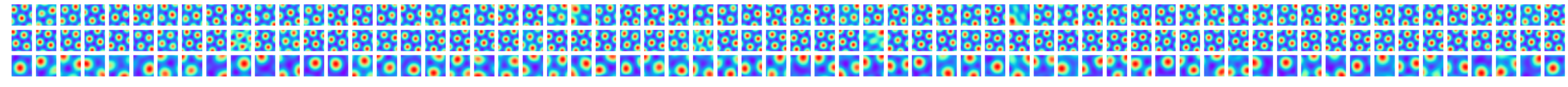}\\
	(d) Block size $= 64$
	\caption{\small Learned patterns of $\vv(\vx)$ with different block sizes. The total number of units is fixed at $192$. Every row shows the learned patterns within the same block.}
	\label{fig: units}
	\end{figure}
	
\paragraph{Multiple blocks or modules are necessary for learning grid patterns of multiple scales.} We further try to fully remove the assumption of blocks or modules; i.e., we learn a single block of $\mB(\theta)$. Fig. \ref{fig:ablation-single} shows the learned neurons and the corresponding autocorrelograms. All the learned neurons share similar large scales, which indicates that the high frequency part of $A(\vx, \vx')$ may not be fitted very well.  %From the autocorrelograms, we observe some clear hexagon grid patterns, while for the rest, the grid scales are probably beyond the scope of the whole area so that we cannot tell from the autocorrelograms easily. 

\begin{figure}[h]
	\centering
	\includegraphics[width=.31\linewidth]{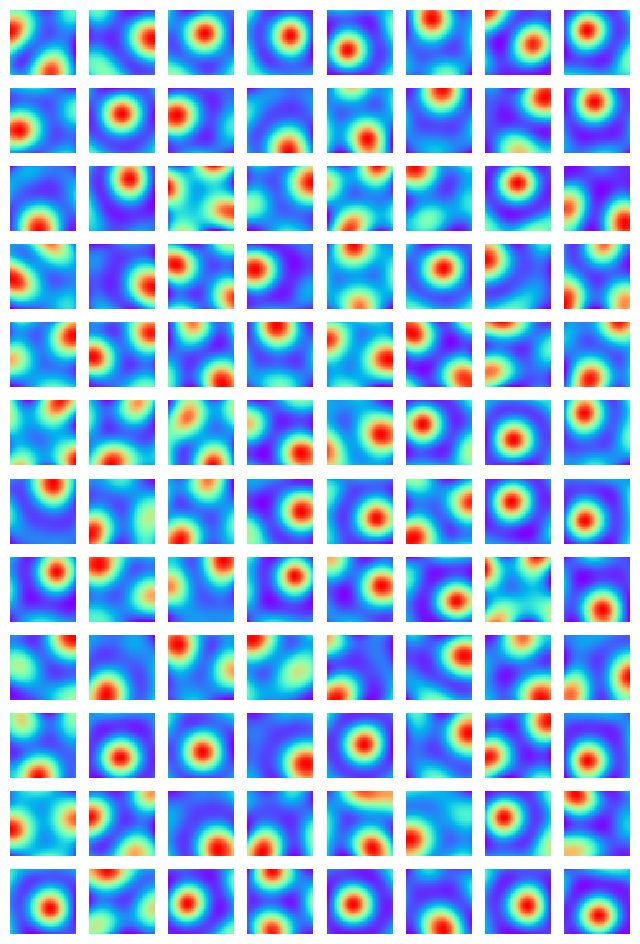} \hspace{2mm}
	\includegraphics[width=.3\linewidth]{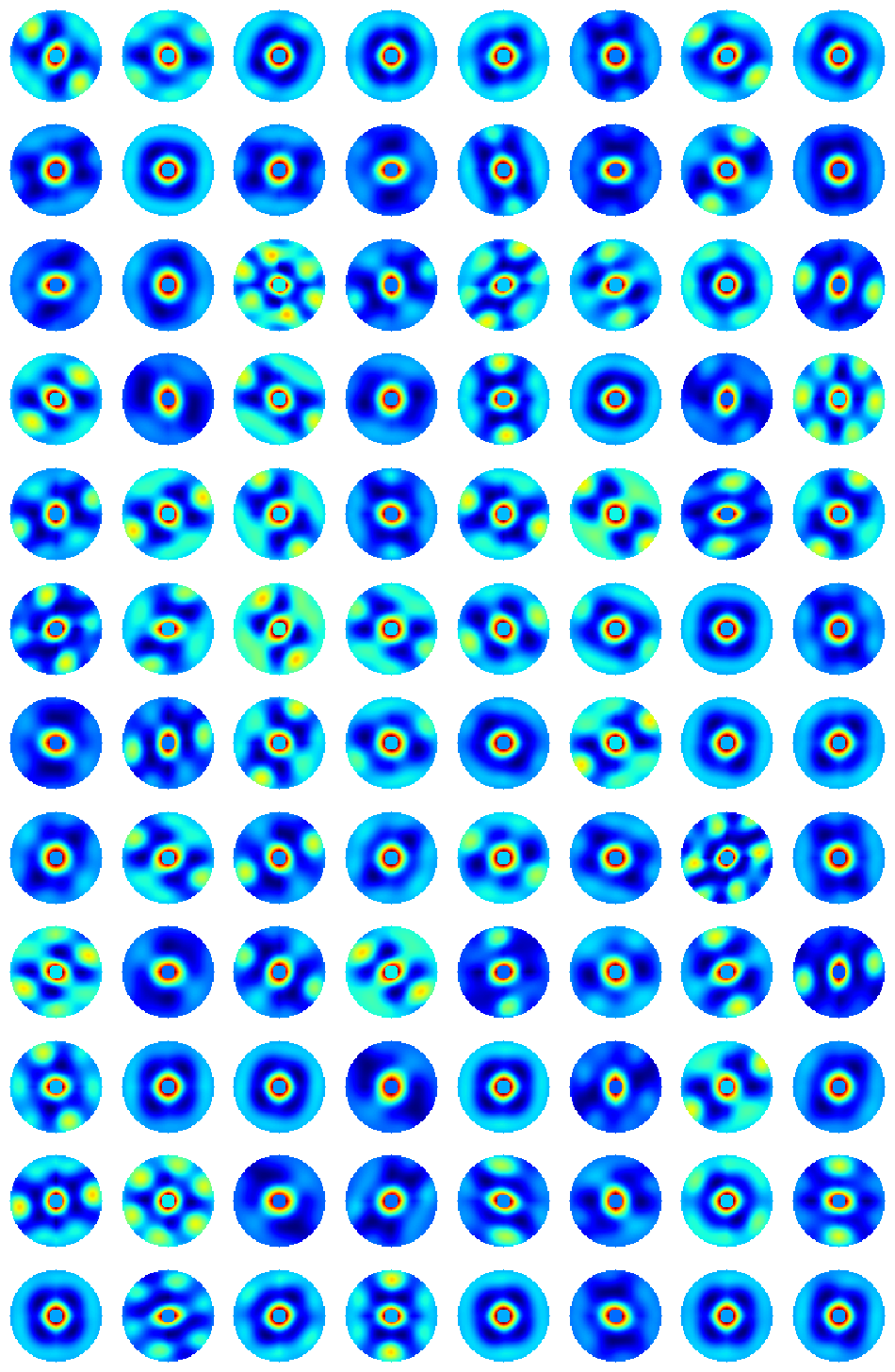}
	\caption{\small {\em Left}: learned neurons with a single block of $\mB(\theta)$. The firing patterns has a single large scale, meaning that the high frequency part of $A(\vx, \vx')$ is not fitted very well. {\em Right}: autocorrelograms of the learned neurons. Some exhibit clear hexagon grid patterns, while the other do not, probably because the scale of those grid patterns are beyond the scope of the whole area.  }
  \label{fig:ablation-single}
\end{figure}

\end{document}